\documentclass[a4paper]{article}
\pdfoutput=1

\usepackage[authoryear]{natbib}
\usepackage[utf8]{inputenc}
\usepackage{amssymb,amsmath}
\usepackage{amsthm}
\usepackage{enumerate}
\usepackage{lmodern}
\usepackage{url}

\newcommand{\emptyans}{[\quad]}
\newcommand{\false}{\operatorname{false}}
\newcommand{\true}{\operatorname{true}}
\newcommand{\N}{\mathbb{N}}
\newcommand{\Z}{\mathbb{Z}}
\newcommand{\Q}{\mathbb{Q}}

\newcommand{\R}{\mathbb{R}}
\newcommand{\msub}{\mathop{/\kern-3pt/}}

\newcommand{\ee}{\mathbf{e}}
\newcommand{\ii}{\mathbf{i}}
\newcommand{\uu}{\mathbf{u}}
\newcommand{\vv}{\mathbf{v}}
\newcommand{\sss}{\mathbf{s}}

\newcommand{\aaa}{\mathbf{a}}

\newcommand{\sgn}{\operatorname{sgn}}
\newcommand{\me}{\mathrel{\dot=}}

\theoremstyle{definition}
\newtheorem{lemma}{Lemma}

\newtheorem{proposition}[lemma]{Proposition}

\theoremstyle{plain}
\newtheorem{corollary}[lemma]{Corollary}
\newtheorem{theorem}[lemma]{Theorem}

\begin{document}
\title{Better Answers to Real Questions}

\author{Marek Košta and Thomas Sturm\\
Max-Planck-Institut für Informatik\\
Saarbrücken, Germany\\
\url{{mkosta,sturm}@mpi-inf.mpg.de}
\and
Andreas Dolzmann\\
Leibniz-Zentrum für Informatik\\
Saarbrücken, Germany\\
\url{andreas.dolzmann@dagstuhl.de}}

\date{January 21, 2015}

\maketitle

\begin{abstract}
  We consider existential problems over the reals. Extended quantifier
  elimination generalizes the concept of regular quantifier elimination by
  providing in addition answers, which are descriptions of possible assignments
  for the quantified variables. Implementations of extended quantifier
  elimination via virtual substitution have been successfully applied to various
  problems in science and engineering. So far, the answers produced by these
  implementations included infinitesimal and infinite numbers, which are hard to
  interpret in practice. We introduce here a post-processing procedure to
  convert, for fixed parameters, all answers into standard real numbers. The
  relevance of our procedure is demonstrated by application of our
  implementation to various examples from the literature, where it significantly
  improves the quality of the results.
\end{abstract}

\section{Introduction}
We consider existential problems over the reals. Extended quantifier elimination
generalizes the concept of regular quantifier elimination by providing in
addition answers, which are descriptions of possible assignments for the
quantified variables~\citep{Weispfenning:94b,Weispfenning:97d}. Implementations
of extended quantifier elimination~\citep{DolzmannSturm:97a,DolzmannSturm:96c}
via virtual
substitution~\citep{LoosWeispfenning:93a,Weispfenning:97b,Weispfenning:1994a,Weispfenning:88a,DolzmanEtAl:1999a}
have been successfully applied to various problems in science and
engineering~\citep{SturmWeispfenning1997,Sturm:99a,Weispfenning2001,SturmWeber:08a,SturmWeber:09a,ErramiSturm:11a,WeberSturm:11a,ErramiEiswirth:13a,Sturm1999}.

So far, the answers produced by these implementations included infinitesimal and
infinite numbers, which are hard to interpret in practice. This has been
explicitly criticized in the literature, e.g., by~\cite{Collard:03a}. In the
present article, we introduce a complete post-processing procedure to convert,
for fixed values of parameters, all answers into standard real numbers. We
furthermore demonstrate the successful application of an implementation of our
method within Redlog~\citep{DolzmannSturm:97a} to a number of extended
quantifier elimination problems from the scientific literature including
computational geometry~\citep{SturmWeispfenning1997}, motion
planning~\citep{Weispfenning2001}, bifurcation analysis for models of genetic
circuits and for mass action~\citep{SturmWeber:08a,SturmWeber:09a}, and sizing
of electrical networks~\citep{Sturm1999}.

The plan of the paper is as follows: In Section~\ref{SE:eqe} we make ourselves
familiar with the concept of extended quantifier elimination. In
Section~\ref{SE:vs} we give an introduction of virtual substitution for extended
quantifier elimination to the extent necessary to understand how nonstandard
values enter the answers and what information is available for fixing them to
standard values. Section~\ref{SE:main} is the technical core; we describe and
prove our method and illustrate it by discussing one example in detail. In
Section~\ref{SE:shift} we revisit \emph{degree shifts}, a successful heuristics
for reducing the degree of quantified variables before their elimination. We
re-interpret these degree shifts as quantifier eliminations by virtual
substitution. This allows us in Section~\ref{SE:ext} to generalize our method to
cover also possible degree shifts during elimination. In Section~\ref{SE:ex} we
revisit examples from the scientific literature where the application of
extended quantifier elimination to various problems from planning, modeling,
science, and engineering had yielded nonstandard answers. In all cases we can
efficiently fix all nonstandard symbols to standard values using our
implementation of the method as it was introduced in Section~\ref{SE:main} and
generalized in Section~\ref{SE:ext}. This significantly improves the quality of
the results from a practical point of view. Finally, in
Section~\ref{SE:conclusions} we summarize our results and discuss possible
extensions of our method.

\section{The Concept of\\ Extended Quantifier Elimination}\label{SE:eqe}
For our purposes here, we restrict ourselves to existential problems
\begin{displaymath}
  \varphi(u_1,\dots,u_m)=\exists x_n\dots\exists x_1\psi(x_1,\dots,x_n,u_1,\dots,u_m)
\end{displaymath}
in the Tarski language $L=(0,1,+,-,\cdot,\leq,<,\geq,>,\neq)$ interpreted in the
Tarski algebra $(\R,0,1,+,-,\cdot,\leq,<,\geq,>,\neq)$. As usual in algebraic
model theory, the symbol ``$=$'' and its interpretation as equality is part of
first-order logic so that it does not occur explicitly in the language $L$.
Without loss of generality, $\psi$ is an $\land$-$\lor$-combination of atomic
constraints, and we agree that all right hand sides of the atomic constraints
are $0$.

\emph{Extended quantifier elimination} applied to $\varphi$ yields an
\emph{extended quantifier elimination result (EQR)}
\begin{displaymath}
  \addtolength{\arraycolsep}{0.25em}
  \left[
    \begin{array}{c|ccc}
      \beta_1(\uu) & x_1=e_{11}(\uu) & \dots & x_n=e_{1n}(\uu)\\
      \vdots & \vdots & \ddots & \vdots\\
      \beta_k(\uu) & x_1=e_{k1}(\uu) & \dots & x_n=e_{kn}(\uu)
    \end{array}
  \right].
\end{displaymath}
The \emph{conditions} $\beta_i(\uu)$ are quantifier-free Tarski formulas such
that $\R\models\varphi\longleftrightarrow\bigvee_{i=1}^k\beta_i$. In other
words, $\bigvee_{i=1}^k\beta_i$ is a regular quantifier elimination result for
$\varphi$, and extended quantifier elimination generalizes regular quantifier
elimination. The \emph{answers} $\ee_i(\uu)$ are terms in an extension language
of the Tarski language. For $\aaa\in\R^m$, if $\varphi(\aaa)$ holds, then at
least one $\beta_i(\aaa)$ holds, and so does $\psi(\ee_{i}(\aaa),\aaa)$. We
agree that ``$\false$'' never occurs as a condition. If $\varphi$ itself is
equivalent to ``$\false$,'' we possibly obtain the empty scheme $\emptyans$.

As an example, consider the input formula
\begin{displaymath}
  \varphi =\exists x\exists y\psi,\quad \psi=a y + 3 x^{2} + 4 x \leq a \land
  x\geq a\geq y.
\end{displaymath}
A possible extended quantifier elimination result for $\varphi$ is given by
\begin{displaymath}
  \addtolength{\arraycolsep}{0.25em}
  \left[\begin{array}{l|ll}
      a \neq 0 \land 4 a + 3 \geq 0 & y = - 3 a - 3 & x = a\\[2ex]
      a\leq 0 \land 3 a^{2} - 3 a - 4 \leq 0 &
      y = a & \displaystyle x = \frac{\sqrt{ - 3 a^{2} + 3 a + 4} - 2}{3\strut}
    \end{array}\right].
\end{displaymath}
From this extended quantifier elimination result we can derive a regular
quantifier elimination result
\begin{displaymath}
  (a \neq 0 \land 4 a + 3 \geq 0)\lor
  (a\leq 0 \land 3 a^{2} - 3 a - 4 \leq0),
\end{displaymath}
which can be simplified to $a\geq0 \lor 3 a^{2} - 3 a - 4 \leq0$. Hence,
$\varphi$ holds if and only if $a\geq\alpha$, where $\alpha\approx - 0.758306$
is the smaller root of $3 a^{2} - 3 a - 4$.

In the extended quantifier elimination result, the first row covers the case
that $-0.75\leq a$ and $a\neq 0$, while the second row covers $\alpha\leq
a\leq0$. Let us consider some particular interpretations of $a$:
\begin{itemize}
\item For $a=2$, the condition in the first row holds and the corresponding
  answers yield $x=2$ and $y=-9$. In fact, these three values satisfy $\psi$.
  The condition in the second row, in contrast, does not hold. If we plug $a=2$
  into the corresponding answers anyway, then we obtain a negative argument for
  the square root, which cannot be interpreted over the reals.
\item For $\alpha<a=-0.7525<-0.75$, the condition in the second row holds and
  the corresponding answers yield $x=\frac{\sqrt{6997} - 800}{1200}$ and
  $y=-0.7525$. Again, these three values satisfy $\psi$. Now the condition in
  the first row does not hold. If we plug $a=-0.7525$ into the corresponding
  answers anyway, then we obtain $x=-0.7525$ and $y=-0.7425$, which does not
  satisfy $\psi$.
\item For $a=-0.5$, both conditions hold and yield two different sets of values
  satisfying $\psi$, viz.~$x=-0.5$, $y=-1.5$ and $x=\frac{\sqrt{7}-4}{6}$, $y=-0.5$,
  respectively.
\item For $a=0$ only the condition in the second row holds, but the answers in
  the first row happen to work as well. This shows that the conditions are
  sufficient but not necessary for the answers to be valid.
\end{itemize}

\section{The Method of Virtual Substitution}\label{SE:vs}
Given $\varphi(\uu)=\exists x\psi(x,\uu)$, we compute a finite \emph{elimination
  set}
\begin{equation}
  \label{EQ:vs}
  E=\bigl\{\dots,\bigl(\gamma(\uu), e(\uu)\bigr), \dots,\bigr\}
  \quad\text{such that}\quad
  \exists x\psi\longleftrightarrow\bigvee_{(\gamma,e)\in
    E}\gamma\land\psi[x\msub e].
\end{equation}
In the elimination set $E$ the $e(\uu)$ are \emph{elimination terms} substituted
for the quantified variable $x$ via a \emph{virtual substitution} $[x\msub e]$.
The $\gamma$ are quantifier-free Tarski formulas serving as \emph{substitution
  guards}. Equation (\ref{EQ:vs}) formally describes regular quantifier
elimination of one quantifier $\exists x$ from $\varphi$. For the elimination of
several quantifiers, one assumes without loss of generality that the formula is
prenex and processes the prenex quantifier block from the inside to the outside.

We are now going to give an idea of the exact shape and computation of
elimination sets that is sufficiently precise to understand our main
contribution here. For a more thorough introduction into the theory of
quantifier elimination by virtual substitution we refer to the original
publications by \cite{Weispfenning:88a,Weispfenning:1994a,Weispfenning:97b},
\cite{LoosWeispfenning:93a}, and \cite{DolzmanEtAl:1999a}. In
Section~\ref{SE:VSweak} we restrict ourselves to input formulas $\varphi$ not
containing any strict inequalities ``$<$,'' ``$>$,'' or ``$\neq$.'' In that
course it will become clear how exactly we derive extended quantifier
elimination from the virtual substitution procedure.

Later, in Section~\ref{SE:VSstrict}, we generalize to formulas containing also
strict inequalities. While with regular quantifier elimination the techniques
used in the course of the generalization to strict inequalities remain
completely transparent to the user, with extended quantifier elimination they
leave visible traces in the answers by possibly introducing certain nonstandard
elements, which do not have a straightforward interpretation over the reals. The
purpose of this paper is to convert these answers to real numbers, given fixed
values for the parameters.

\subsection{Virtual Substitution for Weak Inequalities}
\label{SE:VSweak}
Recall that our constraints are normalized such that their right hand sides are
$0$. Assume that all occurrences of $x$ in $\varphi(\uu)=\exists x\psi(x,\uu)$
are at most quadratic. Consider fixed real interpretations $\uu=\aaa\in\R^m$ for
all parameters. Then all constraints in $\psi(x,\aaa)$ become univariate, and
the set $\{\,c\in\R\mid \R\models\psi(c,\aaa)\,\}$ is a finite union of real
intervals, where the interval endpoints are zeros of the univariate left hand
side polynomials.

Our goal is to include at least one such interval endpoint into our elimination
set $E$: For each constraint
${f_2(\uu)x^2+f_1(\uu)x+f_0(\uu)\mathrel{\varrho}0}$, with a weak relation
$\varrho\in\{=, \leq, \geq\}$ and discriminant $\Delta=f_1^2-4f_2f_0$ we add to
$E$ three pairs $(\gamma, e)$ as follows:
\begin{multline*}
  \biggl(f_2\neq0\land \Delta\geq0,\
  \frac{-f_1+\sqrt{\Delta}}{2f_2}\biggr),\quad\\
  \biggl(f_2\neq0\land \Delta\geq0,\
  \frac{-f_1-\sqrt{\Delta}}{2f_2}\biggr),\quad
  \biggl(f_2=0\land f_1\neq0,\ -\frac{f_0}{f_1}\biggr).
\end{multline*}
In order to obtain a quantifier-free equivalent for $\varphi$, such pairs have
to be plugged into $\varphi$ according to~(\ref{EQ:vs}). To start with, note
that the substitution guards $\gamma$ make the substitution terms $e$ meaningful
by ensuring that denominations are not zero and arguments to square roots are
not negative.

Next, observe that our elimination terms $e$ are not terms in the Tarski
language as they contain division as well as root symbols. It is one central
idea of the virtual substitution approach that the substitution operator does
not map $L$-terms to $L$-terms but atomic $L$-formulas to quantifier-free
$L$-formulas:
\begin{displaymath}
  [x\msub t]:\text{atomic $L$-formulas} \to \text{quantifier-free $L$-formulas}
\end{displaymath}
Note that when there is more than one quantifier, it is crucial to obtain
$L$-formulas in~(\ref{EQ:vs}) in order to be able to proceed.

To give an impression of virtual substitution, we describe here the substitution
$(f=0)[x\msub \frac{g_1+g_2\sqrt{g_3}}{g_4}]$ of a \emph{root expression}
\begin{displaymath}
  \frac{g_1+g_2\sqrt{g_3}}{g_4},\qquad g_i\in\Z[\uu]
\end{displaymath}
into an equation $f=0$, where $f\in\Z[\uu][x]$ of arbitrary degree: It is easy
to see that there are $g_1^*$, $g_2^*$, $g_4^*$ such that
\begin{displaymath}
  f\biggl(\frac{g_1+g_2\sqrt{g_3}}{g_4}\biggr)=\frac{g_1^*+g_2^*\sqrt{g_3}}{g_4^*}
\end{displaymath}
is again a root expression. Using this intermediate result, we transform
\begin{eqnarray*}
  \frac{g_1^*+g_2^*\sqrt{g_3}}{g_4^*}=0
  &\me&
  g_1^{*}+g_2^{*}\sqrt{g_3}=0\\[0.5ex]
  &\me&
  |g_1^{*}|=|g_2^{*}\sqrt{g_3}|\land{}\\
  &&\quad\bigl(\sgn(g_1^*)\neq\sgn(g_2^*)\lor\sgn(g_1^*)=\sgn(g_2^*)=0\bigr)\\[0.5ex]
  &\me& {g_1^*}^2-{g_2^*}^2g_3=0\land g_1^*g_2^*\leq0.
\end{eqnarray*}
Technical details and formal descriptions of virtual substitutions for all our
relations have been given by \cite{Weispfenning:97b}.

Let us now apply these ideas to extended quantifier elimination of several
existential quantifiers via virtual substitution. Given
\begin{displaymath}
  \exists x_n\dots\exists x_1\psi(x_1,\dots, x_n,u_1,\dots,u_m)
\end{displaymath}
our intended result is a scheme as described in Section~\ref{SE:eqe}:
\begin{displaymath}
  \addtolength{\arraycolsep}{0.25em}
  \left[
    \begin{array}{c|ccc}
      \beta_1(\uu) & x_1=e_{11}(\uu) & \dots & x_n=e_{1n}(\uu)\\
      \vdots & \vdots & \ddots & \vdots\\
      \beta_k(\uu) & x_1=e_{k1}(\uu) & \dots & x_n=e_{kn}(\uu)
    \end{array}
  \right].
\end{displaymath}
We successively apply~(\ref{EQ:vs}) to $x_1$, \dots,~$x_n$ using elimination
sets $E_1(x_2,\dots,x_n,\uu)$, \dots, $E_n(\uu)$ to obtain $\beta_1$,
\dots,~$\beta_k$ as follows:
\begin{equation}
  \label{EQ:betadef}
  \bigvee\limits_{(\gamma_n,e_n)\in E_n}\dots\bigvee\limits_{(\gamma_1,e_1)\in E_1}\
  \underbrace{\gamma_n\land\bigl(\dots\land
    \bigl(\gamma_1\land\psi[x_1\msub e_1]\bigr)
    \cdots\bigr)[x_n\msub e_n]}_{\beta_i(\uu)}.
\end{equation}
The index $i$ of $\beta_i$ describes one choice
$\bigl((\gamma_n,e_n),\dots,(\gamma_1,e_1)\bigr)$ from the Cartesian product
$E_n\times\dots\times E_1$. In practice, the $\beta_i$ obtained this way undergo
sophisticated simplification methods such as those described
by~\cite{DolzmannSturm:97c}. Recall from the previous section that $\beta_i$
that become ``$\false$'' are ignored. It is important to understand that for the
computation of the $\beta_i$ we are using exclusively \emph{virtual}
substitution.

The corresponding $\ee_i(\uu)$, in contrast, are obtained from $e_n(\uu)$,
$e_{n-1}(x_n,\uu)$, \dots,~$e_1(x_2,\dots,x_n,\uu)$ via \emph{regular
  back-substitution} of terms in an suitable extension language of $L$:
\begin{equation}
  \label{EQ:backsub}
  e_{i,n}=e_n,\quad
  e_{i,n-1}=e_{n-1}[x_n/e_{i,n}],\quad
  e_{i,1}=e_1[x_1/e_{i,2}]\dots[x_n/e_{i,n}].
\end{equation}

Note once more that the back-substitution possible creates objects like
\begin{displaymath}
  u_1+\frac{3\sqrt{\sqrt{5u_1-u_2}-2}}{2},
\end{displaymath}
which are not $L$-terms or even root expressions. It is thus suitable for
obtaining the $\ee_i$ but not for the $\beta_i$. Vice versa, virtual
substitution requires atomic formulas as input. Thus it is suitable for
obtaining the $\beta_i$ while for the $\ee_i$ it is not applicable at all.
Virtual substitution and regular term substitution are independent concepts,
which complement each other.

\subsection{Virtual Substitution with Strict Inequalities}
\label{SE:VSstrict}
Let us return to the elimination of a single quantifier from $\exists
x\psi(x,\uu)$. Recall from the beginning of the previous Section~\ref{SE:VSweak}
how we considered fixed interpretation $\uu=\aaa\in\R^m$ causing the set
$S=\{\,c\in\R\mid \R\models\psi(c,\aaa)\,\}$ to be a finite union of intervals.

When considering in addition strict inequalities, the intervals in $S$ are
possibly open. Consequently, for a \emph{strict constraint}
\begin{equation}
  \label{EQ:constraint}
  \psi_i(x,\uu)\me
  f_{i2}(\uu)x_i^2+f_{i1}(\uu)x_i+f_{i0}(\uu)\mathrel{\varrho_i}0,\qquad
  \varrho_i\in \{<,>,\neq\},
\end{equation}
contained in $\psi(x,\uu)$ we cannot use the zeros $z_i(\uu)$ of the left hand
side but need a point from \emph{inside} the corresponding interval.

In early versions of virtual substitution methods for the linear case,
\cite{Weispfenning:88a} used \emph{arithmetic means} $\frac12(z_i+z_j)$ for all
pairs $(\psi_i,\psi_j)$ of strict constraints. However, the size of the
elimination set then grows quadratically in the number of constraints, which
turned out to be critical for the practical performance of the
method~\citep{Burhenne:90a,LoosWeispfenning:93a}. For the quadratic case,
observe that expressions
\begin{displaymath}
  \frac12\left(
    \frac{-g_{i1}\pm\sqrt{\Delta_{i}}}{2g_{i2}}+\frac{-g_{j1}\pm{\sqrt{\Delta_j}}}{2g_{j2}}
  \right)
\end{displaymath}
are not root expression of the form
\begin{displaymath}
  \frac{g_{1}^*+g_2^*\sqrt{\Delta^*}}{g_{3}^*},
\end{displaymath}
so that Weispfenning's (\citeyear{Weispfenning:97b}) virtual substitution rules
sketched in the previous Section~\ref{SE:VSweak} cannot be used.

The established approach for strict inequalities uses nonstandard extensions of
$\R$: Let $\varepsilon\in\R^*\supset \R$ be a positive infinitesimal number,
i.e., $0<\varepsilon<x$ for all $0<x\in\R$. Then for a strict constraint as
defined in~(\ref{EQ:constraint}) we use four test points
\begin{displaymath}
  \frac{-f_{i1}\pm\sqrt{\Delta_i}}{2f_{i2}}\pm\varepsilon.
\end{displaymath}
As an optimization, it suffices to consider only upper bounds using
$-\varepsilon$. For solution sets that are unbounded from above we have to add
only one more point $\infty:=1/\varepsilon\in\R^*$ for the entire problem.

For the application of the elimination set as described in~(\ref{EQ:vs}) and
thus for the computation of the $\beta_i$ with extended quantifier elimination,
both $\varepsilon$ and $\infty$ are equivalently translated into the Tarski
language $L$ via virtual substitution. For instance, let $t$ be a standard term.
Then
\begin{multline*}
  (x^3+x^2-x-1<0)[x\msub t-\varepsilon]\me{}\\
  (x^3+x^2-x-1<0\lor(x^3+x^2-x-1=0\land (3x^2+2x-1>0 \lor{}\\
  (3x^2+2x-1=0 \land (6x+2<0 \lor (6x+2=0\land 6 > 0))))))[x\msub t].
\end{multline*}
In a subsequent step, $t$ can be virtually substituted for $x$ as discussed in
the previous section. For understanding the principal idea, notice that
$3x^2+2x-1$, $6x+2$, and $6$ are the first, second, and third derivatives of
$x^3+x^2-x-1$, respectively. For the substitution of $\infty$ we have, e.g.,
\begin{displaymath}
  (ax^2+bx+c<0)[x\msub\infty]\me
  a<0 \lor (a=0 \land b<0)\lor(a=0\land b=0 \land c<0).
\end{displaymath}
Again, precise definitions and proofs have been given by
\cite{Weispfenning:97b}.

When again performing regular back-substitution of terms on the side of the
$\ee_i$, nonstandard symbols cannot be removed but are propagated along the way.
In the final result a single answer $\ee_i$ can even contain several of such
nonstandard symbols. For example, on input of
\begin{equation}
  \label{EQ:epsexin}
  \varphi=\exists x\exists y\psi,\quad \psi\me a y + 3 x^{2} + 4 x < 0
  \land x>y>a
\end{equation}
we obtain a nonstandard extended quantifier elimination result:
\begin{equation}
  \label{EQ:epsex}
  \addtolength{\arraycolsep}{0.25em}
  \left[
    \begin{array}{l|lll}
      a + 4 < 0 &
      \displaystyle
      y = \frac{- a - 3 \varepsilon_{1} - 3 \varepsilon_{2} - 4}{3} &
      \displaystyle
      x = \frac{- a - 3 \varepsilon_{1} - 4}{3}\\[2ex]
      a < 0 \land a + 4 > 0 &
      y = - \varepsilon_{1} - \varepsilon_{2} &
      x = -\varepsilon_{1}
    \end{array}
  \right].
\end{equation}
Given such answers containing nonstandard symbols, it is not hard to
non-constructively prove from the elimination procedure that for fixed real
interpretations of the parameter $a$ there are positive real choices
$\varepsilon_1$, $\varepsilon_2\in\R$ so that the answers satisfy $\psi$.
Infinitesimals introduced at different stages of the elimination are indexed
accordingly. It is noteworthy that they have to be chosen differently in
general: Fixing $a=-2$ in the example, it is easy to see that $\varepsilon_1$
has to be chosen different from $\varepsilon_2$ because otherwise we would
obtain $y=2x$, which together with $a=-2$ does not satisfy $a y + 3 x^{2} + 4 x
< 0$.

Until now users of extended quantifier elimination were left alone with results
as in (\ref{EQ:epsex}). In spite of the difficulties discussed above, there is a
considerable record of applications of extended quantifier elimination in the
literature. We are going to discuss some of these with our examples in
Section~\ref{SE:ex}.

We conclude this section with an important observation that for unfixed
parameters it is not possible in general to determine suitable real choices for
nonstandard symbols:
\begin{proposition}[No Standard Answers for Unfixed Parameters]\label{PR:notparametric}
  Consider the formula $\varphi = \exists x(a < x < 1)$ and the nonstandard
  extended quantifier elimination result
\begin{displaymath}
  \addtolength{\arraycolsep}{0.25em}
  \left[
    \begin{array}{l|l}
      a < 1 &
      \displaystyle
      x = 1 - \varepsilon_1
    \end{array}
  \right].
\end{displaymath}
There is no standard choice $\widetilde{\varepsilon_1}\in\R$ such that
\begin{math}
  \addtolength{\arraycolsep}{0.25em}
  \left[
    \begin{array}{l|l}
      a < 1 &
      \displaystyle
      x = 1 - \widetilde{\varepsilon_1}
    \end{array}
  \right]
\end{math}
is an extended quantifier elimination result for $\varphi$ as well.
\end{proposition}
\begin{proof}
  Assume for a contradiction that $\widetilde{\varepsilon_1}\in\R$ is a suitable
  choice. Then by definition of extended quantifier elimination it follows for
  all $a\in\mathopen]-\infty,1\mathclose[$ that $a < 1 -
  \widetilde{\varepsilon_1} < 1$, in particular $\widetilde{\varepsilon_1} > 0$.
  On the other hand, for $a_0 = 1 - \frac{\widetilde{\varepsilon_1}}{2}$ we have
  $a_0\in\mathopen]-\infty,1\mathclose[$ and $1 - \widetilde{\varepsilon_1} <
  a_0 < 1$, a contradiction.
\end{proof}

\section{Elimination of Nonstandard Symbols\\ from Answers}\label{SE:main}
Given an extended quantifier elimination result and prescribed values for all
parameters, our goal is to compute answers containing only standard real
numbers. For instance, given (\ref{EQ:epsexin}) and (\ref{EQ:epsex}), and fixing
$a=-2$ we are going to obtain
\begin{displaymath}
  \addtolength{\arraycolsep}{0.25em}
  \left[
    \begin{array}{l|lll}
      \true &
      \displaystyle
      y = -\frac{9}{256} &
      \displaystyle
      x = -\frac{1}{32}
    \end{array}
  \right].
\end{displaymath}

From the point-of-view of our method, it makes no difference whether the
parameters are fixed after extended quantifier elimination or in advance. For
the sake of a concise description, we are thus going to restrict to existential
decision problems from now on. Recall that if the regular quantifier elimination
result is ``$\false$,'' then the extended quantifier elimination result is
$\emptyans$, i.e., empty. If the result is ``$\true$,'' then we assume for
simplicity that the extended quantifier elimination result contains only one
row, like
\begin{equation}
  \label{EQ:backanswer}
  \addtolength{\arraycolsep}{0.25em}
  \left[
    \begin{array}{l|lll}
      \true &
      \displaystyle
      x_1 = e_{1,1} &
      \displaystyle
      \dots &
      \displaystyle
      x_n = e_{1,n}
    \end{array}
  \right].
\end{equation}
Recall from (\ref{EQ:backsub}) in Section~\ref{SE:VSweak} that here the
$e_{1,1}$, \dots,~$e_{1,n}$ have been obtained from elimination terms $e_1$,
\dots, $e_{n}$ via regular back-substitution. Our method is going to use not the
back-substituted answers but these original elimination terms. In addition, we
are going to use the substitution guards $\gamma_1$, \dots,~$\gamma_n$
substituted with those elimination terms in (\ref{EQ:betadef}). Hence the input
for our method is not an EQR like in~(\ref{EQ:backanswer}) but a \emph{pre-EQR}
as follows:
\begin{displaymath}
  \label{EQ:nobackanswer}
  \addtolength{\arraycolsep}{0.25em}
  \left[
    \begin{array}{l|rrrr}
      \true &
      x_1 = e_{1}(x_2,\dots,x_n) &
      \displaystyle
      \dots &
      x_{n-1} = e_{n-1}(x_n) &
      x_n = e_{n}()\\
      &\gamma_1(x_2,\dots,x_n)&
      \dots&
      \gamma_{n-1}(x_n)&
      \gamma_n()
    \end{array}
  \right].
\end{displaymath}

\begin{lemma}[Semantics of Virtual Substitution]\label{LE:eps}
  Let $\varphi(x_1,\dots,x_n)$ be a Tarski formula, and let $a_2$,
  \dots,~$a_n\in\R$.
  \begin{enumerate}[(i)]
  \item Assume that $\R\models\varphi[x_1\msub x_2-\varepsilon](a_2,\dots,a_n)$.
    Then there is $\widetilde{\varepsilon_0}\in\R$,
    $0<\widetilde{\varepsilon_0}$, such that for all
    $\widetilde{\varepsilon}\in\R$,
    $0<\widetilde{\varepsilon}<\widetilde{\varepsilon_0}$, we have
    $\R\models\varphi(a_2-\widetilde{\varepsilon}, a_2,\dots,a_n)$.
  \item Assume that $\R\models\varphi[x_1\msub\infty](a_2,\dots,a_n)$. Then
    there is $\widetilde{a_0}\in\R$, such that for all $\widetilde{a_1}\in\R$,
    $\widetilde{a_0}<\widetilde{a_1}$, we have
    $\R\models\varphi(\widetilde{a_1}, a_2,\dots,a_n)$. In particular, the set
    $T=\{\,a\in\R\mid \text{$0<a$ and $\R\models\varphi(a,a_2,\dots,a_n)$}\,\}$
    is unbounded from above.
  \end{enumerate}
\end{lemma}
\begin{proof}
  Consider $L_{\varepsilon}=L\cup\{\varepsilon,\infty\}$. Using the compactness
  theorem for first-order logic, there is a real closed field $\R^*$ where the
  interpretation of $\varepsilon$ is a positive infinitesimal and
  $\infty=\varepsilon^{-1}$. The $L$-restriction of $\R^*$ is a proper
  extension field of $\R$ and, by the Tarski principle, elementary equivalent
  to $\R$. Formally, $\R^*|L\supset\R$ and $\R^*|L\cong\R$.
  \begin{enumerate}[(i)]
  \item Using $\R\cong\R^*|L$ and expanding to $\R^*$ it follows from
    $\R\models\varphi[x_1\msub x_2-\varepsilon](a_2,\dots,a_n)$ that also
    $\R^*\models\varphi[x_1\msub x_2-\varepsilon](a_2,\dots,a_n)$. It has been
    discussed by \cite{LoosWeispfenning:93a} and \cite{Weispfenning:97b} that
    for our $a_2$, \dots,~$a_n\in\R$ the virtual substitution of $\varepsilon$
    has the following property:
    \begin{displaymath}
      \R^*\models\varphi[x_1\msub x_2-\varepsilon](a_2,\dots,a_n)
      \longleftrightarrow
      \varphi[x_1/x_2-\varepsilon](a_2,\dots,a_n).
    \end{displaymath}
    It follows that $\R^*\models\varphi[x_1/x_2-\varepsilon](a_2,\dots,a_n)$. Let
    $n\in\N\setminus\{0\}$. Then we can conclude
    $\R^*\models\psi[x_0/\varepsilon](a_2,\dots,a_n)$, where
    \begin{displaymath}
      \psi=(0<x_0\land nx_0<1\land\varphi)[x_1/x_2-x_0].
    \end{displaymath}
    Now we can generalize $\R^*\models\exists x_0\psi(a_2,\dots,a_n)$. Since
    $\varepsilon$ does not occur anymore, we restrict from $\R^*$ to $\R^*|L$ and
    then use the elementary equivalence to obtain $\R\models\exists
    x_0\psi(a_2,\dots,a_n)$. We have just shown that for any
    $n\in\N\setminus\{0\}$ there exists
    $a_0\in\R$, $0 < a_0 < \frac{1}{n}$, such that $\R\models\varphi(a_2-a_0,
    a_2,\dots,a_n)$. It follows that $\inf S=0$, where
    \begin{displaymath}
      S=\{\,a\in\R\mid \text{$0<a$ and
        $\R\models\varphi(a_2-a,a_2,\dots,a_n)$}\,\}\subseteq \R.
    \end{displaymath}
    On the other hand, $S$ is a semialgebraic set and thus a finite union of
    intervals and points. Hence there is $\widetilde{\varepsilon_0}\in\R$,
    $0<\widetilde{\varepsilon_0}$, such that
    $\mathopen]0,\widetilde{\varepsilon_0}\mathclose[\subseteq S$.
  \item The argument is essentially the same as in (i) above: We conclude that
    $\R^*\models\varphi[x_1\msub\infty](a_2,\dots,a_n)$. According to
    \cite{LoosWeispfenning:93a} and \cite{Weispfenning:97b} we know that
    \begin{displaymath}
      \R^*\models\varphi[x_1\msub\infty](a_2,\dots,a_n)
      \longleftrightarrow
      \varphi[x_1/\infty](a_2,\dots,a_n)
    \end{displaymath}
    so that for $n\in\N$ we can conclude
    $\R^*\models\psi[x_0/\infty](a_2,\dots,a_n)$, where
    \begin{displaymath}
      \psi=(n<x_0\land\varphi)[x_1/x_0].
    \end{displaymath}
    Again, we generalize, restrict, and apply
    elementary equivalence to obtain $\R\models\exists
    x_0\psi(a_2,\dots,a_n)$.
    We thus know that for any $n\in\N$ there exists
    $a_0\in\R$, $n < a_0$, such that $\R\models\varphi(a_0, a_2,\dots,a_n)$. It
    follows that the set $T$ is unbounded from above. On the other hand, $T$
    is a semialgebraic set and
    thus a finite union of intervals and points. Hence there is
    $\widetilde{a_0}\in\R$, such that
    $\mathopen]\widetilde{a_0},\infty\mathclose[\subseteq T$.\qedhere
  \end{enumerate}
\end{proof}

\begin{lemma}\label{LE:zwei}
  Consider a quantifier-free Tarski formula $\psi(x_1,\dots,x_n)$. Assume that
  for each $i\in\{2,\dots,n\}$ we have a root expression
  $\widetilde{e_i}=\frac{a_i + b_i\sqrt{c_i}}{d_i}$ with $a_i$, $b_i$, $c_i$,
  $d_i\in\Z[x_{i+1},\dots,x_n]$. Assume furthermore that
  $\R\models\psi\longrightarrow c_i\geq0\land d_i\neq0$, and let $\alpha_i\in\R$
  be the interpretation of $e_i =
  \widetilde{e_i}[x_{i+1}/\widetilde{e_{i+1}}]\dots[x_n/\widetilde{e_n}]$. Then
  \begin{displaymath}
    \{\,\alpha\in\R\mid\R\models\psi[x_2\msub\widetilde{e_2}]\dots[x_n\msub\widetilde{e_n}](\alpha)\,\}=
    \{\,\alpha\in\R\mid\R\models\psi(\alpha,\alpha_2,\dots,\alpha_n)\,\}.
  \end{displaymath}
\end{lemma}
\begin{proof}
  Consider $L'=L\cup\{\sqrt{\ }, {}^{-1}\}$ and the $L'$-expansion $\R'$ of $\R$
  where $\sqrt{\ }$ and ${}^{-1}$ have the usual semantics if defined and an
  arbitrary but fixed value otherwise. Let $\nu = (f\mathrel{\varrho} 0)$, where
  $f\in\Z[x,u_1,\dots,u_m]$, $\varrho\in\{\leq,<,\geq,>,\neq,=\}$, and let
  $e=\frac{a+b\sqrt{c}}{d}$, where $a$, $b$, $c$, $d\in\Z[u_1,\dots,u_m]$. Let
  $\sss\in\R^m$ such that $c(\sss)\geqslant 0$ and $d(\sss)\neq 0$. Then
  according to \cite{Weispfenning:97b} the following holds for $\sss\in\R^m$:
  \begin{displaymath}
    \R'\models(\nu[x\msub e]
    \longleftrightarrow
    \nu[x/e])(\sss).
  \end{displaymath}
  For our formula $\psi$ and all $\alpha\in\R$ we obtain
  \begin{equation*}
    \R'\models(\psi[x_2\msub\widetilde{e_2}]\dots[x_n\msub\widetilde{e_n}]
    \longleftrightarrow\psi[x_2/{\widetilde{e_2}}]\dots[x_n/{\widetilde{e_n}}]
    \longleftrightarrow\psi
    )(\alpha,\alpha_2,\dots,\alpha_n).
  \end{equation*}
  It follows that
  $\R\models(\psi[x_2\msub\widetilde{e_2}]\dots[x_n\msub\widetilde{e_n}]
  \longleftrightarrow\psi )(\alpha,\alpha_2,\dots,\alpha_n)$ for all
  ${\alpha\in\R}$.
\end{proof}

\begin{lemma}[Commutation of Virtual Substitutions]\label{LE:comm}
  Consider a quantifier-free Tarski formula $\psi(x_1,\dots,x_n)$. Let $e_1 =
  \frac{a_1+b_1\sqrt{c_1}}{d_1}$, with $a_1$, $b_1$, $c_1$, $d_1\in\Z$,
  $c_1\geqslant 0$, $d_1\neq 0$. Furthermore, let $i\in\{2,\dots,n\}$, let $e_i
  =\frac{a_i+b_i\sqrt{c_i}}{d_i}$ with $a_i$, $b_i$, $c_i$,
  $d_i\in\Z[x_{i+1},\dots,x_n]$, and let $\gamma_i(x_{i+1},\dots,x_n)$ be a
  quantifier-free formula such that $\R\models\gamma_i\longrightarrow
  c_i\geqslant 0 \land d_i\neq 0$. Then
  \begin{displaymath}
    \R\models(\gamma_i\land\psi)[x_i\msub e_i][x_1\msub e_1]
    \longleftrightarrow
    (\gamma_i\land\psi)[x_1\msub e_1][x_i\msub e_i].
  \end{displaymath}
\end{lemma}

\begin{proof}
  Let $L'$, $\R'$ be as in the proof of Lemma~\ref{LE:zwei}. Recall the
  equivalence $\R'\models(\nu[x\msub e] \longleftrightarrow \nu[x/e])(\sss)$ and
  observe that, on the other hand, if $\R\models(c_i<0)(\sss)$ or $\R\models(d_i
  = 0)(\sss)$ for $\sss\in\R^{n-i}$, then
  $\R\models(\gamma_i\land\psi\longleftrightarrow\false)(\sss)$. It follows
  that
  \begin{eqnarray*}
    \R' \models (\gamma_i\land\psi)[x_i\msub e_i][x_1\msub e_1]
    & \longleftrightarrow &
    (\gamma_i\land\psi)[x_i/e_i][x_1/e_1]\\
    & \longleftrightarrow &
    (\gamma_i\land\psi)[x_1/e_1][x_i/e_i]\\
    & \longleftrightarrow &
    (\gamma_i\land\psi)[x_1\msub e_1][x_i\msub e_i].
  \end{eqnarray*}
  Since virtual substitution eliminates all occurrences of $\sqrt{\ }$ and
  ${}^{-1}$, we can conclude that
  $\R\models (\gamma_i\land\psi)[x_i\msub e_i][x_1\msub e_1]
    \longleftrightarrow
    (\gamma_i\land\psi)[x_1\msub e_1][x_i\msub e_i]$.
\end{proof}

\begin{theorem}[Computation of Standard Answers]\label{TH:main}
  Consider a closed Tarski formula $\varphi = \exists x_n\dots\exists
  x_1\psi(x_1,\dots,x_n)$. Assume that
  \begin{displaymath}
    \addtolength{\arraycolsep}{0.25em}
    \left[
      \begin{array}{l|rrr}
        \true &
        x_1 = e_1 &
        \displaystyle
        \dots &
        x_n = e_n\\
        &
        \gamma_1 &
        \dots &
        \gamma_n
      \end{array}
    \right]
  \end{displaymath}
  is a pre-EQR for $\varphi$ such that each $e_i$ is of one of the following
  forms:
  \begin{enumerate}[(a)]
  \item a root expression $\frac{a + b\sqrt{c}}{d}$, where $a$, $b$, $c$,
    $d\in\Z[x_{i+1},\dots,x_n]$,
  \item $\infty$,
  \item $x_{i+1} - \varepsilon$.
  \end{enumerate}
  Then we can compute root expressions $\widetilde{e_1}$,
  \dots,~$\widetilde{e_n}$ meeting the specification (a) above and
  $\widetilde{\gamma_1}$, \dots,~$\widetilde{\gamma_n}$ such that the following
  is a pre-EQR for $\varphi$ as well:
  \begin{equation}\label{EQ:corrpreans}
    \addtolength{\arraycolsep}{0.25em}
    \left[
      \begin{array}{l|rrr}
        \true &
        \displaystyle
        x_1 = \widetilde{e_1} &
        \displaystyle
        \dots &
        \displaystyle
        x_n = \widetilde{e_n}\\[0.5ex]
        &
        \widetilde{\gamma_1} &
        \dots &
        \widetilde{\gamma_n}
      \end{array}
    \right].
  \end{equation}
\end{theorem}

\begin{proof}
  For the sake of the proof, we are going to show that in addition to the
  required $\widetilde{e_1}$, \dots,~$\widetilde{e_n}$ and
  $\widetilde{\gamma_1}$, \dots,~$\widetilde{\gamma_n}$ we can compute real
  algebraic numbers $\alpha_1$, \dots,~$\alpha_n$ for the values of
  $\widetilde{e_1}$, \dots,~$\widetilde{e_n}$ after back-substitution. We
  represent these real algebraic numbers as pairs of univariate defining
  polynomials and open isolating intervals with rational bounds. Given
  $k\in\{1,\dots,n\}$, it suffices to show that from
  \begin{equation*}
    \addtolength{\arraycolsep}{0.2em}
    \left[
      \begin{array}{l|rrrrrr}
        \true &
        x_1 = e_1 &
        \dots &
        x_k = e_k &
        x_{k+1} = \widetilde{e_{k+1}} &
        \dots &
        x_n = \widetilde{e_n}\\[0.5ex]
        &
        \gamma_1 &
        \dots &
        \gamma_k &
        \widetilde{\gamma_{k+1}} &
        \dots &
        \widetilde{\gamma_n}
      \end{array}
    \right]
  \end{equation*}
  and $\alpha_{k+1}$, \dots,~$\alpha_n$ we can compute suitable
  $\widetilde{e_k}$, $\widetilde{\gamma_k}$, and $\alpha_k$. Define
  \begin{eqnarray*}
    \varphi_k(x_k,\dots,x_n)&=&(\gamma_{k-1}\land\dots\land\gamma_1\land\psi)[x_1\msub
    e_1]\dots[x_{k-1}\msub e_{k-1}],\\
    \varphi_{k+1}(x_{k+1},\dots,x_n)&=&(\gamma_k\land\varphi_k)[x_k\msub e_k],
  \end{eqnarray*}
  and observe that
  \begin{equation}\label{EQ:indpreans}
    \addtolength{\arraycolsep}{0.25em}
    \left[
      \begin{array}{l|rrrr}
        \true &
        x_k = e_k &
        x_{k+1} = \widetilde{e_{k+1}} &
        \dots &
        x_n = \widetilde{e_n}\\[0.5ex]
        &
        \gamma_k &
        \widetilde{\gamma_{k+1}} &
        \dots &
        \widetilde{\gamma_n}
      \end{array}
    \right]
  \end{equation}
  and
  \begin{equation}\label{EQ:indpreans2}
    \addtolength{\arraycolsep}{0.25em}
    \left[
      \begin{array}{l|rrr}
        \true &
        x_{k+1} = \widetilde{e_{k+1}} &
        \dots &
        x_n = \widetilde{e_n}\\[0.5ex]
        &
        \widetilde{\gamma_{k+1}} &
        \dots &
        \widetilde{\gamma_n}
      \end{array}
    \right]
  \end{equation}
  are pre-EQRs for $\exists x_n\dots\exists x_{k}\varphi_k$ and $\exists
  x_n\dots\exists x_{k+1}\varphi_{k+1}$, respectively. On the basis of these
  definitions it is sufficient for our proof to compute suitable
  $\widetilde{e_k}$, $\widetilde{\gamma_k}$, and $\alpha_k$ such that
  \begin{equation*}
    \addtolength{\arraycolsep}{0.25em}
    \left[
      \begin{array}{l|rrrr}
        \true &
        x_k = \widetilde{e_k} &
        x_{k+1} = \widetilde{e_{k+1}} &
        \dots &
        x_n = \widetilde{e_n}\\[0.5ex]
        &
        \widetilde{\gamma_k} &
        \widetilde{\gamma_{k+1}} &
        \dots &
        \widetilde{\gamma_n}
      \end{array}
    \right]
  \end{equation*}
  is a pre-EQR for $\exists x_n\dots\exists x_{k}\varphi_k$ as well. We
  define furthermore
  \begin{eqnarray*}
    \xi(x_k,\dots,x_n)&=&\widetilde{\gamma_n}\land\dots\land\widetilde{\gamma_{k+1}}\land\varphi_k,\\
    \xi'(x_k)&=&\xi[x_{k+1}\msub\widetilde{e_{k+1}}]\dots[x_n\msub\widetilde{e_n}].
  \end{eqnarray*}
  Lemma~\ref{LE:zwei} applied to the quantifier-free formula $\xi$, the root
  expressions $\widetilde{e_{k+1}}$, \dots,~$\widetilde{e_n}$, and the real
  algebraic numbers $\alpha_{k+1}$, \dots,~$\alpha_n$ yields
  \begin{equation}\label{EQ:equalsets}
    \{\,r\in\R\mid\R\models\xi'(r)\,\}
    =
    \{\,r\in\R\mid\R\models\xi(r,\alpha_{k+1},\dots,\alpha_n)\,\}.
  \end{equation}
  We distinguish three cases depending on the type of $e_k$.
  \begin{enumerate}[(a)]
  \item We have $e_k = \frac{a_k + b_k\sqrt{c_k}}{d_k}$, and $\gamma_k$ equals
    $d_k \neq 0\land c_k \geqslant 0$. We set $\widetilde{e_k}=e_k$ and
    $\widetilde{\gamma_k}=\gamma_k$. Since (\ref{EQ:indpreans2}) is a pre-EQR
    for $\exists x_n\dots\exists x_{k+1}\varphi_{k+1}$ and $\alpha_{k+1}$,
    \dots,~$\alpha_{n}$ correspond to the values of $\widetilde{e_{k+1}}$,
    \dots,~$\widetilde{e_{n}}$ after back-substitution, we have
    $\R\models\varphi_{k+1}(\alpha_{k+1},\dots,\alpha_n)$. It follows that in
    particular $\R\models\gamma_k(\alpha_{k+1},\dots,\alpha_n)$ and furthermore
    $\R\models d_k(\alpha_{k+1},\dots,\alpha_n)\neq 0$ and $\R\models
    c_k(\alpha_{k+1},\dots,\alpha_n)\geqslant 0$. This allows us to compute
    $\alpha_k=\widetilde{e_k}(\alpha_{k+1},\dots,\alpha_k)$ from $a_k$, $b_k$,
    $c_k$, $d_k$, $\alpha_{k+1}$, \dots,~$\alpha_n$.
  \item We have $e_k=\infty$, and $\gamma_k$ is ``$\true$.'' Since
    (\ref{EQ:indpreans}) is a pre-EQR for formula $\exists x_n\dots\exists
    x_k\varphi_k$ we have
    \begin{equation*}
      \R\models\xi[x_k\msub\infty][x_{k+1}\msub\widetilde{e_{k+1}}]\dots[x_n\msub\widetilde{e_n}].
    \end{equation*}
    Using the fact that $\alpha_{k+1}$, \dots,~$\alpha_n$ are real algebraic
    numbers corresponding to the values of $\widetilde{e_{k+1}}$,
    \dots,~$\widetilde{e_n}$ after back-substitution we conclude that
    \begin{equation*}
      \R\models\xi[x_k\msub\infty](\alpha_{k+1},\dots,\alpha_n).
    \end{equation*}
    Lemma~\ref{LE:eps}(ii) now guarantees that the set
    $\{\,r\in\R\mid\R\models\xi(r,\alpha_{k+1},\dots,\alpha_n)\,\}$ is unbounded
    from above. Thus by~(\ref{EQ:equalsets}) the set
    $\{\,r\in\R\mid\R\models\xi'(r)\,\}$ is unbounded from above as well. Using
    well-known bounds~\citep{Akritas:09a} on the roots of the univariate
    polynomials contained in $\xi'$, we compute a sufficiently large
    $\frac{p}{q}\in\Q$ satisfying $\xi'$. We set $\widetilde{e_k} = \frac{p +
      0\sqrt{0}}{q}$ and construct a corresponding real algebraic number
    $\alpha_k$, and we set $\widetilde{\gamma_k}=\true$. Then
    $\R\models(\widetilde{\gamma_k}\land\xi')[x_k\msub\widetilde{e_k}]$, and
    $n-k$ applications of Lemma~\ref{LE:comm} yield
    \begin{displaymath}
      \R\models
      (\widetilde{\gamma_n}\land\dots\land\widetilde{\gamma_k}\land\varphi_k)
      [x_k\msub\widetilde{e_k}][x_{k+1}\msub\widetilde{e_{k+1}}]\dots[x_n\msub\widetilde{e_n}].
    \end{displaymath}
  \item We have $e_k=x_{k+1}-\varepsilon$, and $\gamma_k$ is ``$\true$.''
    Similarly to case (b), we observe that
    (\ref{EQ:indpreans}) is a pre-EQR for $\exists
    x_n\dots\exists x_k\varphi_k$ and obtain
    \begin{equation*}
      \R\models\xi[x_k\msub x_{k+1}-\varepsilon][x_{k+1}\msub\widetilde{e_{k+1}}]\dots[x_n\msub\widetilde{e_n}],
    \end{equation*}
    and conclude that $\R\models\xi[x_k\msub
    x_{k+1}-\varepsilon](\alpha_{k+1},\dots,\alpha_n)$. Lemma~\ref{LE:eps}(i)
    now guarantees the existence of some $\widetilde{\varepsilon_0}\in\R$,
    $0<\widetilde{\varepsilon_0}$, such that
    \begin{displaymath}
      \R\models\xi(\alpha_{k+1}-\widetilde{\varepsilon},\alpha_{k+1},\dots,\alpha_n)
      \quad\text{for}\quad0<\widetilde{\varepsilon}<\widetilde{\varepsilon_0}.
    \end{displaymath}
    By~(\ref{EQ:equalsets}) it follows that
    $\R\models\xi'(\alpha_{k+1}-\widetilde{\varepsilon})$ for all
    $0<\widetilde{\varepsilon}<\widetilde{\varepsilon_0}$. Therefore, after
    finitely many refinements of the isolating interval
    $\bigl]\frac{p}{q},u\bigr[$ of $\alpha_{k+1}$ we obtain
    $\R\models\xi'\bigl(\frac{p}{q}\bigr)$. We set
    $\widetilde{e_k}=\frac{p+0\sqrt{0}}{q}$ and construct a corresponding real
    algebraic number, and we set $\widetilde{\gamma_k}=\true$. Exactly as in
    case (b),
    $\R\models(\widetilde{\gamma_k}\land\xi')[x_k\msub\widetilde{e_k}]$, and
    $n-k$ applications of Lemma~\ref{LE:comm} yield
    \begin{displaymath}
      \R\models
      (\widetilde{\gamma_n}\land\dots\land\widetilde{\gamma_k}\land\varphi_k)
      [x_k\msub\widetilde{e_k}][x_{k+1}\msub\widetilde{e_{k+1}}]\dots[x_n\msub\widetilde{e_n}].
    \end{displaymath}
    Note that instead of using the lower bound $\frac{p}{q}$ one can
    heuristically try and find a satisfying integer.\qedhere
  \end{enumerate}
\end{proof}

A careful inspection of our proof reveals that in all cases
$\widetilde{\gamma_k}=\gamma_k$. However, this is going to change in
Corollary~\ref{CO:main}, which generalizes our theorem.

Notice that the constructive proof of Theorem~\ref{TH:main} suggests to
recompute the intermediate quantifier elimination results $\varphi_k$. In
practice, there are arguments for saving these $\varphi_k$ during the quantifier
elimination run. Consider, e.g., the following common optimization: Whenever
some $\varphi_k$ heuristically simplifies to a disjunction
$\varphi_{k,1}\lor\dots\lor\varphi_{k,s}$, then the virtual substitution
procedure would treat each $\varphi_{k,j}$ separately, i.e., like originating
from several elimination set elements. In general, in the course of the
application of Theorem~\ref{TH:main} such transformations cannot be
reconstructed exclusively from the pre-EQR.

To illustrate the theorem we revisit our example given in (\ref{EQ:epsexin}) on
page~\pageref{EQ:epsexin} for the choice $a=-2$. In that case $\psi$ in
(\ref{EQ:epsexin}) specializes to
\begin{equation*}
  \psi\me-2y + 3 x^{2} + 4 x < 0 \land x>y>-2.
\end{equation*}
For our theorem we have to consider the following pre-EQR corresponding to
the specialization of the EQR (\ref{EQ:epsex}) to $a=-2$:
\begin{equation}\label{EQ:expreans}
  \addtolength{\arraycolsep}{0.25em}
  \left[
    \begin{array}{l|rrrr}
      \true & y = x - \varepsilon & x = h-\varepsilon & h = 0\\
            & \true & \true                & \true
    \end{array}
  \right].
\end{equation}
Notice the introduction of an artificial variable $h$ to meet the requirement of
the theorem that infinitesimals occur only in expressions of the form $x_i =
x_{i+1}-\varepsilon$.

To apply the theorem to the pre-EQR (\ref{EQ:expreans}), we consider
\begin{eqnarray*}
  \varphi_1(y, x, h) & \me & -2y + 3 x^{2} + 4 x < 0 \land x>y>-2,\\
  \varphi_2(x, h) & \me & \varphi_1[y\msub x-\varepsilon]\\
  & \me & x + 2 > 0 \land 3 x^{2} + 2 x < 0,\\
  \varphi_3(h) & \me & \varphi_2[x\msub h-\varepsilon]\\
  & \me & h + 2 > 0 \land (h = 0 \lor 3 h^{2} + 2 h < 0).
\end{eqnarray*}
As in the theorem, we proceed from the right to the left, i.e., our first step
is fixing $h$ and computing a respective algebraic number $\alpha_h$. Since
$h=0$, we are in case (a). Now
\begin{equation*}
  \addtolength{\arraycolsep}{0.25em}
  \left[
    \begin{array}{l|r}
      \true & h = 0\\
            & \true
    \end{array}
  \right]
\end{equation*}
is a pre-EQR for $\exists h\varphi_3$. Therefore, $\alpha_h$ is the root of
the polynomial $h$ in the interval $\mathopen]-1, 1\mathclose[$, i.e.,
$\alpha_h$ is the rational number 0.

We continue with $x = h - \varepsilon$, which is case (c). Now
\begin{equation*}
  \addtolength{\arraycolsep}{0.25em}
  \left[
    \begin{array}{l|rrr}
      \true & x = h-\varepsilon & h = 0\\
            &  \true            & \true
    \end{array}
  \right]
\end{equation*}
is a pre-EQR for $\exists x\exists h\varphi_2$. Lemma~\ref{LE:eps}(i) ensures
that there exists $\widetilde{\varepsilon_0}\in\R$,
$0<\widetilde{\varepsilon_0}$ such that for all $\widetilde{\varepsilon}\in\R$,
$0<\widetilde{\varepsilon}<\widetilde{\varepsilon_0}$, we have
$\R\models\varphi_2(\alpha_h-\widetilde{\varepsilon}, \alpha_h)$. Refining
$\alpha_h$ we compute that $\R\models\varphi_2(\alpha_x, \alpha_h)$, where
$\alpha_x$ is the root of $32x+1$ in the interval $\mathopen]-\frac{1}{16},
\frac{1}{16}\mathclose[$, i.e., $\alpha_x$ is the rational number
$-\frac{1}{32}$.

Finally consider $y = x - \varepsilon$, which is again case (c). Now
\begin{equation*}
  \addtolength{\arraycolsep}{0.25em}
  \left[
    \begin{array}{l|rrrr}
      \true & y = x - \varepsilon  & x = -\frac{1}{32} & h = 0\\
            &  \true               & \true            & \true
    \end{array}
  \right]
\end{equation*}
is a pre-EQR for $\exists y\exists x\exists h\varphi_1$.
Lemma~\ref{LE:eps}(i) ensures that there exists
$\widetilde{\varepsilon_0}\in\R$, $0<\widetilde{\varepsilon_0}$ such that for
all $\widetilde{\varepsilon}\in\R$,
$0<\widetilde{\varepsilon}<\widetilde{\varepsilon_0}$, we have
$\R\models\varphi_1(\alpha_x-\widetilde{\varepsilon}, \alpha_x, \alpha_h)$.
Refining $\alpha_x$ we compute that $\R\models\varphi_1(\alpha_y, \alpha_x,
\alpha_h)$, where $\alpha_y$ is the root of $256y + 9$ in the interval
$\mathopen]-\frac{10}{256}, \frac{10}{256}\mathclose[$, i.e., $\alpha_y$ is the
rational number $-\frac{9}{256}$. To conclude we state that
\begin{equation*}
  \addtolength{\arraycolsep}{0.25em}
  \left[
    \begin{array}{l|rrrr}
      \true & y = -\frac{9}{256}   & x = -\frac{1}{32} & h = 0\\
            &  \true               & \true             & \true
    \end{array}
  \right]
\end{equation*}
is a pre-EQR for $\exists y\exists x\exists h\varphi_1$, which does not
contain any nonstandard symbols. Since $h$ does not occur in $\varphi_1=\psi$,
\begin{equation*}
  \addtolength{\arraycolsep}{0.25em}
  \left[
    \begin{array}{l|rrr}
      \true & y = -\frac{9}{256}   & x = -\frac{1}{32}\\
            &  \true               & \true
    \end{array}
  \right]
\end{equation*}
is a pre-EQR for $\exists y\exists x\psi$.

Finally, note that all quantified variables have to be present in a pre-EQR
before Theorem~\ref{TH:main} can be applied. This has a consequence, which we
illustrate on an example. Consider a valid sentence $\exists x\exists a(x < a)$.
An extended quantifier elimination result containing a nonstandard symbol for
this formula is
\begin{displaymath}
  \addtolength{\arraycolsep}{0.25em}
  \left[
    \begin{array}{c|c}
      \true & a = \infty
    \end{array}
  \right].
\end{displaymath}
Since $x$ does not occur in this result, it can be ``freely chosen,'' i.e., the
result is independent of the value of $x$. Put another way, this means that
\begin{displaymath}
  \addtolength{\arraycolsep}{0.25em}
  \left[
    \begin{array}{c|cc}
      \true & x = t & a = \infty
    \end{array}
  \right]
\end{displaymath}
is an extended quantifier elimination result for any standard term $t$ as well.
This degree of freedom disappears when computing standard answers in the
following sense: The term $t$ has to be fixed before the computation of standard
answers. Fixing $t=2$ and using Theorem~\ref{TH:main} we obtain a standard
EQR
\begin{displaymath}
  \addtolength{\arraycolsep}{0.25em}
  \left[
    \begin{array}{c|cc}
      \true & x = 2 & a = 3
    \end{array}
  \right].
\end{displaymath}
The computed standard answer for $a$ depends on this choice of $t$, i.e., it is
possibly invalid for other choices. Fixing for example $x=4$, we see that $a=3$
is not admissible anymore, because substituting these terms into $(x<a)$ yields
``$\false$.'' To compute a standard term for $a$ when $x=4$, we have to start
with
\begin{displaymath}
  \addtolength{\arraycolsep}{0.25em}
  \left[
    \begin{array}{c|cc}
      \true & x = 4 & a = \infty
    \end{array}
  \right],
\end{displaymath}
and apply Theorem~\ref{TH:main} again.

\section{Degree Shifts by Virtual Substitution}\label{SE:shift}
We have already discussed that the feasibility of the virtual substitution
method strongly depends on the degrees of the quantified variables. Among the
heuristics for decreasing these degrees there is an observation, which was
essentially made already by \cite{Weispfenning:97b}, and which was refined and
named \emph{degree shift} by \cite{DSW:98}. The following lemma restates the
result by \citeauthor{DSW:98}:
\begin{lemma}[Degree Shift]\label{LE:shift}
  Consider a quantifier-free Tarski formula $\psi$. Let $g$ be the GCD of all
  exponents of $x$ in $\psi$. We divide all exponents of $x$ in $\psi$ by $g$
  yielding $\psi'$. If $g$ is odd, we have
  $\exists x\psi\longleftrightarrow\exists x\psi'$, if $g$ is even we have
  $\exists x\psi\longleftrightarrow \exists x(x\geq0\land\psi')$. For $g>1$ this
  reduces the degree of $x$ in $\psi$. In order to obtain larger GCDs and
  hence a better degree reduction, we may in advance ``adjust'' the degree $n>0$
  of $x$ in polynomials of the form $x^nf$, where $x$ does not occur in $f$ as
  follows: In equations and disequations, $n$ may be equivalently replaced by
  any $m>0$. In ordering inequalities we may choose any $m>0$ of the same parity
  as $n$.\qed
\end{lemma}

We now want to reanalyze this result as a special case of virtual substitution.
For this, we have to slightly generalize the framework by introducing
\emph{shadow quantifiers}. Recall that we are considering existential problems
of the form
\begin{displaymath}
  \varphi(u_1,\dots,u_m)=\exists x_n\dots\exists x_1\psi(x_1,\dots,x_n,u_1,\dots,u_m).
\end{displaymath}
As a first step we now switch to the equivalent problem
\begin{displaymath}
  \hat\varphi(u_1,\dots,u_m)=\exists \hat x_n\exists x_n\dots
  \exists \hat x_1\exists x_1\psi(x_1,\dots,x_n,u_1,\dots,u_m),
\end{displaymath}
where $\{\hat x_1,\dots,\hat
x_n\}\cap\{x_1,\dots,x_n,u_1,\dots,u_m\}=\emptyset$. That is, the variables
$\hat x_i$ do not occur in $\psi$. Consequently, proceeding with the elimination
as discussed in Section~\ref{SE:vs}, each shadow quantifier $\exists\hat x_i$
imposes a trivial elimination problem.

Strictly following the virtual substitution framework, one would not simply drop
those quantifiers $\exists\hat x_i$ but eliminate them via trivial elimination
sets like $\{(\true,0)\}$. Notice that one cannot use $\emptyset$ as an
elimination set here because $\bigvee\emptyset=\false$. Furthermore, from the
point of view of extended quantifier elimination the use of $\{(\true,0)\}$
formally provides answers also for $\hat x_1$, \dots,~$\hat x_n$.

Consider now w.l.o.g.~the elimination of $\exists x_1$, and assume that we are
in the situation of Lemma~\ref{LE:shift}, where $g>1$ is the GCD of
the, possibly adjusted, degrees of $x_1$ in $\psi$. We use an elimination set
that depends on the parity of $g$:
\begin{displaymath}
  E=\Bigl\{\Bigl(\gamma, \sqrt[g]{\hat x_1}\Bigr)\Bigr\},
\end{displaymath}
where $\gamma$ is $x_1\geq0$ if $g$ is even and ``$\true$'' otherwise. Of
course, we have to define a suitable virtual substitution of $\sqrt[g]{\hat
  x_1}$ for $x_1$ within $\psi$:
\begin{displaymath}
  \Biggl(
    \sum_{j=1}^ka_jx_1^j\mathrel{\varrho}0
  \Biggr)\Biggl[x_1\msub\sqrt[g]{\hat x_1}\Biggr]=
  \Biggl(
  \sum_{j=1}^ka_j\hat
  x_1^{\left\lfloor\frac{\max(j,g)}{g}\right\rfloor}\mathrel{\varrho}0
  \Biggr),
\end{displaymath}
where $a_1$, \dots,~$a_k\in\Z[x_2,\dots,x_n,u_1,\dots,u_m]$ and
$\varrho\in\{=,\leq,<,\geq,>,\neq\}$. The floor function is applied to make the
definition complete; with our elimination sets we will always have divisibility
$g\mid\max(j,g)$. The $\max$ operator takes care of possible degree adjustments
made for the computation of $E$.

Observe that in contrast to the elimination sets studied so far we introduce
here a variable $\hat x_1$ which was not present in $\psi$ before. That variable
is bound by shadow quantifier $\exists\hat x_1$. Intuitively, for the
elimination of $\exists\hat x_1\exists x_1$ we switch from one hard plus one
trivial elimination step to two nontrivial elimination steps.

The termination of quantifier elimination with shadow quantifiers follows from
the termination of the underlying quantifier elimination method plus the fact
that there are only finitely many shadow quantifiers for each regular
quantifier.

To keep the notation simple, we will in the sequel not formally introduce shadow
quantifiers for all quantifiers. Instead, our procedure will silently assume
their presence whenever it performs a degree shift. In the corresponding
pre-EQR this can be recognized by assignments of the form $x_i =
\sqrt[g]{\hat x_i}$, which cannot come into existence otherwise.

\section{Generalization and Extensions of the Method}\label{SE:ext}
In this section we generalize Theorem~\ref{TH:main} to admit more general
pre-EQRs as input. Furthermore, we discuss heuristics for obtaining rational
numbers or even integers instead of root expressions in our standard answers.

\begin{corollary}[Generalized Computation of Standard Answers]\label{CO:main}
  Consider a closed Tarski formula $\varphi = \exists x_n\dots\exists
  x_1\psi(x_1,\dots,x_n)$. Assume that the following is a pre-EQR for
  $\varphi$:
  \begin{displaymath}
    \addtolength{\arraycolsep}{0.25em}
    \left[
      \begin{array}{l|rrr}
        \true &
        x_1 = e_1 &
        \displaystyle
        \dots &
        x_n = e_n\\
        &
        \gamma_1 &
        \dots &
        \gamma_n
      \end{array}
    \right]
  \end{displaymath}
  such that each $e_i$ is of one of the following forms:
  \begin{enumerate}[(a)]
  \item $\frac{a + b\sqrt{c}}{d}$, where $a$, $b$, $c$,
    $d\in\Z[x_{i+1},\dots,x_n]$,
  \item $\frac{a + b\sqrt{c}}{d}\pm\varepsilon$, where $a$, $b$, $c$,
    $d\in\Z[x_{i+1},\dots,x_n]$,
  \item $\pm\infty$,
  \item $\sqrt[g]{x_{i+1}}$, where $g\in\N\setminus\{0\}$,
  \end{enumerate}
  where as usual ``$\pm$'' denotes ``$+$'' or ``$-$.'' Then we can compute root
  expressions $\widetilde{e_1}$, \dots,~$\widetilde{e_n}$ each meeting either
  the specification (a) or the specification (d) above and
  $\widetilde{\gamma_1}$, \dots,~$\widetilde{\gamma_n}$ such that the following
  is a pre-EQR for $\varphi$ as well:
  \begin{displaymath}
    \addtolength{\arraycolsep}{0.25em}
    \left[
      \begin{array}{l|rrr}
        \true &
        \displaystyle
        x_1 = \widetilde{e_1} &
        \displaystyle
        \dots &
        \displaystyle
        x_n = \widetilde{e_n}\\[0.5ex]
        &
        \widetilde{\gamma_1} &
        \dots &
        \widetilde{\gamma_n}
      \end{array}
    \right].
  \end{displaymath}
\end{corollary}

\begin{proof}
  From a theoretical point of view, the treatment of $e_i = \frac{a +
    b\sqrt{c}}{d}-\varepsilon$ can be reduced to Theorem~\ref{TH:main} via the
  introduction of artificial variables as we did for our example in
  (\ref{EQ:expreans}) above. From a practical point of view, it is not hard to
  see how to algorithmically treat such expressions directly. Notice that then
  our corrected answer $\widetilde{e_i}$ will generally be a rational number. As
  already mentioned in the proof of Theorem~\ref{TH:main}(c), one might
  heuristically even find an integer. In both cases the possibly non-trivial
  guard $\gamma_i$ has to be replaced by $\widetilde{\gamma_i}=\true$.

  Next, the treatment of $\frac{a + b\sqrt{c}}{d}+\varepsilon$ and $-\infty$ in
  analogy to $\frac{a + b\sqrt{c}}{d}-\varepsilon$ and $\infty$, respectively,
  is straightforward.

  Finally, having obtained algebraic numbers for $x_{i+1}$, \dots,~$x_n$, one
  can compute an algebraic number also for $\sqrt[g]{x_{i+1}}$ with
  $g\in\N\setminus\{0\}$.
\end{proof}

The proofs for both Theorem~\ref{TH:main} and Corollary~\ref{CO:main} are
constructive. Recall that the ordering of the variables within the given pre-EQR
is such that quantifier elimination has taken place from the left to the right,
while the construction of the standard answers proceeds from the right to the
left.

Consider the computation of $\widetilde{e_k}$ for some $e_k$. Here, the
quantifier elimination direction mentioned above has played an important role in
our proofs: Although $e_{k+1}$, \dots,~$e_{n}$ have been replaced with
$\widetilde{e_{k+1}}$, \dots,~$\widetilde{e_{n}}$, the expression $e_k$ is still
valid. Taking that idea a bit further, we may replace $e_k$ with any valid
expression without affecting the validity of either $e_1$, \dots,~$e_{k-1}$ or
$\widetilde{e_{k+1}}$, \dots,~$\widetilde{e_{n}}$.

In fact, it is sometimes possible to convert a root expression $e_k$ into a
rational number or even an integer as follows: Before processing $e_k$, we check
whether changing it to one of $e_k\pm\varepsilon$ yields a valid pre-EQR for
$\varphi$ as well. This can be done by means of the virtual substitution
\begin{multline*}
  (\widetilde{\gamma_n}\land\dots\land\widetilde{\gamma_{k+1}}\land
  \gamma_k\land\dots\land \gamma_1\land\psi)\\
  [x_1\msub e_1]\dots[x_{k-1}\msub e_{k-1}][x_k\msub
  e_k\pm\varepsilon][x_{k+1}\msub\widetilde{e_{k+1}}]\dots[x_n\msub\widetilde{e_n}].
\end{multline*}
In the positive case, we process $e_k\pm\varepsilon$ instead of $e_k$. In terms
of the proofs of Theorem~\ref{TH:main} and Corollary~\ref{CO:main} this leads to
the cases (c) and (b), respectively, where we generally obtain a rational
solutions $\widetilde{e_k}$ and heuristically even integers.

Finally, it is quite helpful in general to recognize rational numbers among all
occurring real algebraic numbers. This holds in particular for the final
$\alpha_n$, \dots,~$\alpha_1$, as they correspond to the values of the
back-substituted $\widetilde{e_n}$, \dots,~$\widetilde{e_1}$, which may be
complicated nested root expressions. For this one can use the following lemma.
\begin{lemma}[Rational Algebraic Numbers]\label{LE:ratnum}
  Consider a real algebraic number
  $\alpha=\bigl(a_nx^n+\dots+a_0,\left]l,u\right[\bigr)$, where $a_0$,
  \dots,~$a_n\in\Z$, $a_0>0$, $l$, $u\in\Q$, $l>0$. Assume furthermore that
  ${\left]\frac{a_0}{u},\frac{a_0}{l}\right[}\cap{\Z}=\{z\}$. Then $\alpha\in\Q$
  if and only if $\alpha=\frac{a_0}{z}$.
\end{lemma}
\begin{proof}
  Let $\alpha\in\Q$. From $l>0$ it follows that $\alpha>0$, say
  $\alpha=\frac{p}{q}$, where $p$, $q\in\Z$, $p>0$, $q>0$. This admits the
  following factorization:
  \begin{displaymath}
    q\cdot\sum_{i=0}^na_ix^i=(qx-p)\cdot\sum_{i=0}^{n-1}a_ix^i.
  \end{displaymath}
  It follows that $p\mid a_0$, say $pp'=a_0$, and we obtain
  $\alpha=\frac{pp'}{qp'}=\frac{a_0}{qp'}$. On the other hand,
  $l<\frac{a_0}{qp'}<u$, which is equivalent to
  $\frac{a_0}{u}<qp'<\frac{a_0}{l}$, and it follows that $qp'=z$. Together
  $\alpha=\frac{a_0}{qp'}=\frac{a_0}{z}$. The converse implication is obvious.
\end{proof}
The lemma can be straightforwardly generalized to arbitrary intervals
$\left]l,u\right[$.

\section{Implementation and Application Examples}\label{SE:ex}
We have implemented our method in Redlog, which is a part of the computer
algebra system Reduce. Reduce is freely available under a modified BSD
license.\footnote{\url{http://reduce-algebra.sourceforge.net}} Technically, our
implementation is an extension of Redlog's extended quantifier elimination
\texttt{rlqea} by a switch \texttt{rlqestdans}, which toggles the computation of
standard answers.

In the following subsections we are going to revisit a number of applications of
extended quantifier elimination that have been documented in the scientific
literature. In each case we are going to briefly explain the underlying problem,
recompute the solutions with nonstandard answers, and finally compute solutions
with standard answers using our approach as described in this article.

Since Redlog is very actively developed and improved, and the considered
applications date back up to more than 15 years, the nonstandard answers
obtained here partly differ from those reported in the literature. Of course, in
such cases both variants are correct.

All computations have been carried out with the CSL version of Reduce, revision
2465, using 4~GB RAM on a 2.4~GHz Intel Xeon E5-4640 running 64~bit Debian
Linux~7.3.

\subsection{Computational Geometry}
Besides many standard problems from computational geometry,
\cite{SturmWeispfenning1997} consider in their Example~10 the reconstruction of
a cuboid wireframe from a photography taken from the origin along the $x_3$-axis
with a lens of focal length 5.

The answers obtained by extended quantifier elimination is going to describe
vectors $\ee_1$, $\ee_2$, $\ee_3\in\R^3$ generating the cuboid together with a
vector $\vv\in\R^3$ describing its translation from the origin. The input
formula, which contains in addition points $\ii\in\R^2$ on the camera sensor,
contains 15 quantifiers:
\begin{displaymath}
  \exists \ee_1
  \exists \ee_2
  \exists \ee_3
  \exists \vv
  \forall \ii
  \bigl( (\iota'\longleftrightarrow\pi_0)
  \land
  \exists k (59k\vv = (100, 200, 295k + 295))
  \bigr).
\end{displaymath}
The formula $\iota'(\ee_1,\ee_2,\ee_3,\vv,\ii)$, which has been obtained by
regular quantifier elimination earlier, generically describes that a point $\ii$
lies in the image of a cuboid generated by $\ee_1$, $\ee_2$, $\ee_3$, and
translated by $\vv$. The formula $\pi_0(\ii)$ is a quantifier-free description
of one concrete image. The remaining part of the input formula fixes $\ii =
\bigl(\frac{100}{59}, \frac{200}{59}\bigr)$ to be the image of the origin of the
cuboid.

Extended quantifier elimination yields ``$\true$'' if and only if $\pi_0$ is a
picture of a cuboid at all. In the positive case, the answers will provide
suitable vectors $\ee_1$, $\ee_2$, $\ee_3$, and $\vv$.

For $\pi_0$ as considered by \cite{SturmWeispfenning1997} in Example~10, the
extended quantifier elimination yields ``$\true$'' together with the following
nonstandard answers:
\begin{align*}
  \ee_{1} &= \left(5\infty_1, \frac{7\infty_1}{2}, \frac{5\infty_1}{2}\right),&
  \ee_{2} &= \left(\infty_1, 2\infty_1, \frac{-24\infty_1}{5}\right),\\[1ex]
  \ee_{3} &= \left(\frac{-109\infty_1}{65}, \frac{53\infty_1}{26},
    \frac{\infty_1}{2}\right),&
  \vv &= \left(5\infty_1, 10\infty_1,\frac{59\infty_1 + 20}{4}\right).
\intertext{Our method fixes $\infty_1=1$, which yields the following standard answers:}
  \ee_{1} &= \left(5, \frac{7}{2}, \frac{5}{2}\right),&
  \ee_{2} &= \left(1, 2, -\frac{24}{5}\right),\\[1ex]
  \ee_{3} &= \left(-\frac{109}{65}, \frac{53}{26}, \frac{1}{2}\right),&
  \vv &= \left(5, 10, \frac{79}{4}\right).
\end{align*}
The entire computation takes 189\,s, of which the computation of the standard
answers takes less than 1\,ms.

\subsection{Motion Planning}
\cite{Weispfenning2001} has studied motion planning problems in dimension two.
Both the object to be moved and the free space between given obstacles are
semilinear sets. Extended quantifier elimination is used to decide whether a
geometrical object can be moved from an initial to a final destination in at
most $n$ moves, where the trajectory of each move is a line segment. In the
positive case, the answers describe the coordinates $\uu_1$,
\dots,~$\uu_n\in\R^2$ of the object after each of the $n$ moves. Accordingly,
the input formulas contain $2n$ variables in the prenex existential block.

We have applied our answer correction to three of the examples discussed by
\cite{Weispfenning2001}. For the concrete input formulas and pictures of the
scenery we refer to that publication.

For Example~6.4, we obtain the following nonstandard answers:
\begin{displaymath}
  \uu_{1} = \left(5-\varepsilon_{1}, 5-\varepsilon_{1}\right),\quad
  \uu_{2} = \left(5-\varepsilon_{1}, \frac{-2\varepsilon_{1}+1}{2}\right),\quad
  \uu_{3} = \left(9, \frac{9}{2}\right).
\end{displaymath}
Our method fixes $\varepsilon_{1}=\frac{3}{16}$, which yields the following
standard answers:
\begin{displaymath}
  \uu_{1} = \left(\frac{77}{16}, \frac{77}{16}\right),\quad
  \uu_{2} = \left(\frac{77}{16}, \frac{5}{16}\right),\quad
  \uu_{3} = \left(9, \frac{9}{2}\right).
\end{displaymath}
The entire computation takes 60\,ms, of which the computation of the standard
answers takes less than 1\,ms. Tables~\ref{TAB:mp1} and \ref{TAB:mp2} summarize
these results along with the two other examples.

\begin{table}[t]
  \centering
  \addtolength{\tabcolsep}{0.1em}
  \renewcommand{\arraystretch}{1.25}
  \begin{tabular}{llllrr}
    \hline
    Example &
    \multicolumn{1}{c}{$\uu_1$} &
    \multicolumn{1}{c}{$\uu_2$} &
    \multicolumn{1}{c}{$\uu_3$} &
    \multicolumn{1}{c}{time}\\
    \hline\\[-2.1ex]
    6.4 &
    $\left(5-\varepsilon_{1}, 5-\varepsilon_{1}\right)$ &
    $\left(5-\varepsilon_{1}, \frac{-2\varepsilon_{1}+1}{2}\right)$&
    $\left(9, \frac{9}{2}\right)$ &
    0.06\,s\\
    6.8 &
    $\left(0, 3+\varepsilon_{1}\right)$ & $\left(3-\varepsilon_{1}, 6\right)$ &
    $\left(7, 6\right)$ & 9.5\,s\\
    6.9 &
    $\left(0, 3+\varepsilon_{1}\right)$ & $\left(3-\varepsilon_{1}, 6\right)$ &
    & 0.28\,s\\
    \hline
  \end{tabular}
  \caption{Summary of nonstandard answers and computation times for motion planning
    examples considered by~\cite{Weispfenning2001}.\label{TAB:mp1}}
\end{table}

\begin{table}[t]
  \centering
  \addtolength{\tabcolsep}{0.9em}
  \renewcommand{\arraystretch}{1.25}
  \begin{tabular}{lrlll}
    \hline
    Example &
    $\varepsilon_1$ &
    \multicolumn{1}{c}{$\uu_1$} &
    \multicolumn{1}{c}{$\uu_2$} &
    \multicolumn{1}{c}{$\uu_3$}\\
    \hline\\[-2.1ex]
    6.4 &
    $\frac{3}{16}$ &
    $\left(\frac{77}{16}, \frac{77}{16}\right)$ &
    $\left(\frac{77}{16}, \frac{5}{16}\right)$ &
    $\left(9, \frac{9}{2}\right)$\\
    6.8 &
    1 &
    $\left(0, 4\right)$ &
    $\left(2, 6\right)$ &
    $\left(7, 6\right)$\\
    6.9 &
    1 &
    $\left(0, 4\right)$ &
    $\left(2, 6\right)$ &\\
    \hline
  \end{tabular}
  \caption{Summary of standard answers for motion planning examples considered by
    \cite{Weispfenning2001}. In all cases the
    time spent for the computation of $\varepsilon_1$ was less than 1\,ms.
\label{TAB:mp2}}
\end{table}

\subsection{Models of Genetic Circuits}
Recently, symbolic methods for the identification of Hopf bifurcations in vector
fields arising from biological networks or chemical reaction networks have
received considerable attention in the literature
\citep{SturmWeber:08a,SturmWeber:09a,ErramiSturm:11a,WeberSturm:11a,ErramiEiswirth:13a}.
Given a polynomial vector field, \cite{El-KahouiWeber:00a} introduced a method,
which automatically generates first-order Tarski formulas describing the
existence of a Hopf bifurcation in terms of the parameters. Then real quantifier
elimination is applied to obtain corresponding necessary and sufficient
conditions. For efficiency reasons, one often existentially quantifies all
parameters and applies extended quantifier elimination. In the positive case,
the answers provide one set of parameter values giving rise to a Hopf
bifurcation.

Based on models introduced by \cite{BoulierEtAl2007}, \cite{SturmWeber:08a} and
\cite{SturmWeber:09a} used the approach sketched above to automatically derive
the existence of Hopf bifurcations for the gene regulatory network controlling
the circadian clock of a certain unicellular green alga. The input formula is
\begin{eqnarray*}
  \lefteqn{\underline{\exists}(0 < \vv_{1} \land 0 < \vv_{3} \land 0 < \vv_{2}
  \land 0 < \vartheta \land {0 < \gamma_{0}}
  \land 0 < \mu \land 0 < \delta\land 0<\alpha}\\
    &&\quad{}\land {\vartheta\cdot
      (\gamma_{0}-\vv_{1}-\vv_{1} \vv_{3}^9) = 0} \land
    \lambda_{1}\vv_{1}+\gamma_{0} \mu-\vv_{2} = 0    \\
    &&\quad{}\land9 \alpha
    (\gamma_{0}-\vv_{1}-\vv_{1} \vv_{3}^9)+\delta (\vv_{2}-\vv_{3}) = 0
    \land
    0 < \vartheta \delta+\vartheta \vv_{3}^9 \delta +9 \lambda_{1}
    \vartheta \vv_{1} \vv_{3}^8 \delta
    \\
    &&\quad{}\land162 \vartheta \vv_{3}^{17} \alpha
    \vv_{1}+162 \vartheta \alpha \vv_{1} \vv_{3}^8 +162 \alpha \vv_{1}
    \vv_{3}^8 \delta+\vartheta+2 \vartheta \vv_{3}^9 \delta +\vartheta^2
    \vv_{3}^{18} \delta \\
    &&\qquad{}+\vartheta \vv_{3}^9+2 \vartheta \delta+81 \alpha
    \vv_{1} \vv_{3}^8 \vartheta \delta +81 \alpha \vv_{1} \vv_{3}^{17} \vartheta
    \delta+\delta^2+\vartheta \delta^2+\vartheta^2 \delta+\vartheta^2\\
    &&\qquad{}+2\vartheta^2 \vv_{3}^9+\vartheta^2 \vv_{3}^{18} +6561 \alpha^2 \vv_{1}^2
    \vv_{3}^{16}+2 \vartheta^2 \vv_{3}^9 \delta+\delta +81 \alpha \vv_{1}
    \vv_{3}^8+\vartheta \vv_{3}^9 \delta^2\\
    &&\qquad{}-9 \lambda_{1} \vartheta \vv_{1}
    \vv_{3}^8 \delta = 0),
  \end{eqnarray*}
for which we obtain the following nonstandard answers:
\begin{eqnarray*}
  \gamma_0 & = & \textstyle\frac{8\sqrt[9]{\infty_2}^{19}\infty_3 + 16\sqrt[9]{\infty_2}^{10}
                 \infty_3 + 8\sqrt[9]{\infty_2}\infty_3}{
                 729\infty_1^2\infty_2^3 + 1458\infty_1^2\infty_2^2 +
                 729\infty_1^2\infty_2 - 486\infty_1\infty_2^2\infty_3
                 - 486\infty_1\infty_2\infty_3 + 9\infty_2\infty_3^2},\\[2ex]
  \mu & = & \textstyle\frac{729\infty_1^2\infty_2^3 + 1458\infty_1^2\infty_2^2 + 729
            \infty_1^2\infty_2 - 486\infty_1\infty_2^2\infty_3
            - 486\infty_1\infty_2\infty_3 + \infty_2\infty_3^2 - 8\infty_3^2}
            {8\infty_2^2\infty_3 + 16\infty_2\infty_3 + 8\infty_3},\\[2ex]
  \vartheta & = & {{- 6561\infty_1^4\infty_2^4 - 26244\infty_1^4\infty_2^3 -
                  39366\infty_1^4\infty_2^2 - 26244\infty_1^4\infty_2
                  - 6561\infty_1^4 + 4374\infty_1^3\infty_2^3\infty_3
                  \atop
                  + 13122\infty_1^3\infty_2^2
                  \infty_3
                  + 13122\infty_1^3\infty_2\infty_3
                  + 4374\infty_1^3\infty_3 - 54\infty_1\infty_2\infty_3^3 - 54\infty_1\infty_3^3 +
                  \infty_3^4}
                  \over
                  {{{\scriptstyle 6561\infty_1^4\infty_2^5 + 32805\infty_1^4\infty_2^4
                  + 65610\infty_1^4\infty_2^3 + 65610\infty_1^4\infty_2^2
                  + 32805\infty_1^4\infty_2 + 6561\infty_1^4
                  \atop\scriptstyle
                  - 8748\infty_1^3\infty_2^4\infty_3 - 34992\infty_1^3\infty_2^3\infty_3
                  - 52488\infty_1^3
                  \infty_2^2\infty_3 - 34992\infty_1^3\infty_2\infty_3 - 8748\infty_1^3\infty_3}
                  \atop\scriptstyle
                  +3078\infty_1^2\infty_2^3\infty_3^2
                  + 9234\infty_1^2
                  \infty_2^2\infty_3^2 + 9234\infty_1^2\infty_2\infty_3^2 + 3078
                  \infty_1^2\infty_3^2 - 108\infty_1\infty_2^2\infty_3^3}
                  \atop
                  - 216\infty_1\infty_2\infty_3^3 - 108\infty_1\infty_3^3 + \infty_2
                  \infty_3^4 + \infty_3^4}},\\[2ex]
  \mathbf{v}_1 & = & \textstyle\frac{8\sqrt[9]{\infty_2}^{10}\infty_3 + 8\sqrt[9]{\infty_2}\infty_3}
                     {729\infty_1^2\infty_2^3 + 1458\infty_1^2\infty_2^2 + 729\infty_1^2\infty_2 -
                     486\infty_1\infty_2^2\infty_3
                     - 486\infty_1\infty_2\infty_3 + 9\infty_2\infty_3^2},\\[2ex]
  \mathbf{v}_2 & = & \sqrt[9]{\infty_2},\quad
                     \mathbf{v}_3  =  \sqrt[9]{\infty_2},\quad
                     \alpha  =  \infty_1,\quad
                     \delta  =  1,\quad
                     \lambda_1  =  \infty_3.
\end{eqnarray*}

Our method fixes $\infty_1=1$, $\infty_2=9$, and $\infty_3=87481$, which yields
the following standard solution:
\begin{align*}
  \gamma_0 & =  \frac{69984800\cdot \sqrt[9]{9}}{616061191401},&
  \mu & =  \frac{3827162521}{69984800},&
  \vartheta & =  \frac{7652917261}{76056937210},\\[1ex]
  \mathbf{v}_1 & =  \frac{6998480\cdot \sqrt[9]{9}}{616061191401},&
  \mathbf{v}_2 & =  \sqrt[9]{9},&
  \mathbf{v}_3 & =  \sqrt[9]{9},\\[1ex]
  \alpha & =  1,&
  \delta & =  1,&
  \lambda_1 & =  87481.\\
\end{align*}
The entire computation takes 370\,ms, of which the computation of the standard
answers takes 140\,ms.

\subsection{Mass Action Systems}
We are now going to discuss another example about Hopf bifurcation. This time,
the considered system is a chemical reaction system, viz.~the famous and
well-studied phosphofructokinase reaction. It has been firstly analyzed with
symbolic methods by
\citeauthor{GatermannEtAl2005}~(\citeyear{GatermannEtAl2005}, Example~2.1). We
adopt here the first-order formulation discussed by \cite{SturmWeber:08a} and
\cite{SturmWeber:09a} following the approach sketched in the previous
subsection.

We obtain nonstandard answers of the following form:
\begin{align*}
  k_{21} & = K_{21}(\infty_1, \infty_2, \infty_3, \infty_4, \varepsilon_1),&
  k_{34} & = \infty_1,&
  k_{43} & = \infty_2,\\
  k_{46} & = K_{46}(\infty_1, \infty_2, \infty_3, \infty_4, \infty_5, \varepsilon_1),&
  k_{64} & =  \infty_5,&
  k_{65} & =  \infty_3,\\
  k_{56} & = K_{56}(\infty_1, \infty_2, \infty_3, \infty_4, \infty_5, \varepsilon_1),&
  \mathbf{v}_{1} & =  \frac{\infty_2\infty_4}{\infty_1},\\
  \mathbf{v}_{2} & = V_{2}(\infty_1, \infty_2, \infty_3, \infty_4, \infty_5, \varepsilon_1),&
  \mathbf{v}_{3} & =  \infty_4.
\end{align*}
The nonstandard terms $K_{21}$, \dots,~$K_{56}$, $V_2$ are so large that
we cannot explicitly display them here. To give an idea, $K_{46}$ would fill
more than 16 pages in this document.

Our method fixes $\infty_1=\infty_2=\infty_3=\infty_4=1$, $\infty_5=20$, and
$\varepsilon_1=2(\sqrt{2}-1)$, which yields the following standard solution:
\begin{align*}
  k_{21} & = 3,&
  k_{34} & =  1,&
  k_{43} & =  1,\\
  k_{46} & =  \frac{\sqrt{3457} + 1}{8},&
  k_{64} & =  20,&
  k_{65} & =  1,\\
  k_{56} & =  \frac{-\sqrt{3457} + 159}{6},&
  \mathbf{v}_{1} & =  1,\\
  \mathbf{v}_{2} & =  \frac{- \sqrt{3457} + 159}{24},&
  \mathbf{v}_{3} & =  1.
\end{align*}
The entire computation takes 13.2\,s, of which the computation of the standard
answers takes 0.1\,s.

\subsection{Sizing of Electrical Networks}
\citeauthor{Sturm1999}~(\citeyear{Sturm1999}, Section~5) has applied generic
quantifier elimination to the sizing of a BJT amplifier. Description of a
circuit is given as a set of operating point equations $E_1$ and a set of AC
conditions $E_2$. For the concrete equations we refer to the mentioned
publication. The system $E_1\land E_2$ has to be solved w.r.t. the main
variables $M=\{r_1,\dots,r_8,c_3\}$ in terms of parameter variables
$P=\{v_{\textrm{cc}},a_{\textrm{high}},a_{\textrm{low}},p,z_{\textrm{in}},z_{\textrm{out}}\}$.
Fixing values of the parameters to
\begin{displaymath}
  v_{\textrm{cc}} = 3,\ a_{\textrm{high}} = 3,\ a_{\textrm{low}} = 2,\ p = 12,\
  z_{\textrm{in}} = 5,\ z_{\textrm{out}} = 5,
\end{displaymath}
the answers contain one nonstandard term:
\begin{align*}
  &r_{1} = \frac{4457058395}{5180672},\quad
    r_{2} = \frac{4457058395}{2590336},
  &r_{3} = -\frac{4457058395}{1295168},&\\[1em]
  &r_{4} = -\frac{4182864929375836679}{128066211840000},\quad
    r_{5} = \infty_{1},
  &r_{6} = \frac{282999424999}{804520000},&\\[1em]
  &r_{7} = 5,\quad
    r_{8} = \frac{25509595605337086755}{20836792295619328},
  &c_{3} = \frac{647584}{13371175185}.&
\end{align*}
Our method fixes $\infty_1=1$, which yields a standard answer for $r_5$. The
entire computation takes less than 2\,ms, of which the computation of the
standard answer takes less than 1\,ms.

\subsection{A Linear Feasibility Example}
\citeauthor{KorovinEtAl2009}~(\citeyear{KorovinEtAl2009}, Section~9) have
considered a small linear existential problem to demonstrate the difference
between their conflict resolution method and the Fourier--Motzkin elimination
method. The following are nonstandard answers for that problem computed by
Redlog:
\begin{align*}
x_{1} &= -\frac{8}{13},&
x_{2} &= \frac{1 - 65\varepsilon_{1}}{65},&
x_{3} = \frac{-14 + 13\varepsilon_{2}}{13},\\[1em]
x_{4} &= \frac{-302 -195\varepsilon_{1} + 65\varepsilon_{2}}{130},&
x_{5} &= \frac{-30 + 26 \varepsilon_{2}}{39}.&
\end{align*}
Our method fixes $\varepsilon_1=\frac{1}{65}$ and $\varepsilon_2=\frac{1}{13}$,
which yields the following standard answers:
\begin{displaymath}
  x_{1} = -\frac{8}{13},\quad
  x_{2} = 0,\quad
  x_{3} = -1,\quad
  x_{4} = -\frac{30}{13},\quad
  x_{5} = -\frac{28}{39}.
\end{displaymath}
The entire computation takes 3\,ms, of which the computation of the standard
answers takes less than 1\,ms.

\section{Conclusions and Future Work}\label{SE:conclusions}
We have introduced extended quantifier elimination as a general concept, and
focused on virtual substitution as one possible method for its realization.
Successful applications of extended quantifier elimination via virtual
substitution have been documented in the literature over the past two decades.
One problem there was that the answers obtained via virtual substitution in
general contain nonstandard symbols, which can be hard to interpret. For fixed
parameters the present work resolves this issue by providing a complete
post-processing method for fixing all answers to standard real numbers. We have
implemented our method, and applied it to various extended quantifier
elimination problems from the literature. In these experiments we have generally
obtained standard answers that are meaningful in terms of the modeled problems.
In most cases our post-processing method does not significantly contribute to
the overall computation time. It is noteworthy that our method is compatible
with our recent work on combining virtual substitution with learning
techniques~\citep{KorovinEtAl:2014a}.

Recall from our discussion in Section~\ref{SE:ext} on finding integer and
rational answers that there is often a considerable degree of freedom in the
choice of standard answers. In future this can be further exploited in many
interesting ways: For instance, using extended quantifier elimination methods as
a theory solver in the context of Satisfiability Modulo Theory (SMT)
solving~\citep{NieuwenhuisEtAl:06a}, in particular when combining several
theories in a Nelson--Oppen~(\citeyear{NelsonOppen:1979a}) style, one is
specifically interested in avoiding identical answers for different variables.
Alternatively, one can try to identify certain variables, which might be
interesting in certain contexts. As another option, there is only a small step
to automatically generating for a given pre-EQR code for an interactive
procedure that suggests ranges and finds possible choices for certain variables
in cooperation with the user. In cooperation projects with researchers from the
sciences we have had the experience that those researchers often have a
surprisingly precise idea about reasonable choices for certain variables.

A theoretically way more challenging step would be the generalization of our
method to the parametric case. Recall that Proposition~\ref{PR:notparametric}
has shown that it is not possible in general to determine real standard values
for infinitesimals and infinities without fixing values for the parameters
beforehand. Nevertheless, it might well be possible to devise on the basis of
our work here a complete method for \emph{symbolically} replacing nonstandard
symbols with standard terms. In the example in
Proposition~\ref{PR:notparametric} the infinitesimal $\varepsilon_1$ could be
replaced e.g.~with $\frac{1 - a}{2}$ yielding a standard extended quantifier
elimination result.

\section*{Acknowledgments}
This research was supported in part by the German Transregional Collaborative
Research Center SFB/TR 14 AVACS and by the DFG/ANR Programme Blanc Project STU
483/2-1 SMArT.


\begin{thebibliography}{30}
\providecommand{\natexlab}[1]{#1}
\providecommand{\url}[1]{\texttt{#1}}
\expandafter\ifx\csname urlstyle\endcsname\relax
  \providecommand{\doi}[1]{doi: #1}\else
  \providecommand{\doi}{doi: \begingroup \urlstyle{rm}\Url}\fi

\bibitem[Akritas(2009)]{Akritas:09a}
Alkiviadis~G. Akritas.
\newblock Linear and quadratic complexity bounds on the values of the positive
  roots of polynomials.
\newblock \emph{Journal of Universal Computer Science}, 15\penalty0
  (3):\penalty0 523--537, 2009.

\bibitem[Boulier et~al.(2007)Boulier, Lefranc, Lemaire, Morant, and
  Ürgüplü]{BoulierEtAl2007}
François Boulier, Marc Lefranc, François Lemaire, Pierre-Emmanuel Morant, and
  Aslı Ürgüplü.
\newblock On proving the absence of oscillations in models of genetic circuits.
\newblock In Hirokazu Anai, Katsuhisa Horimoto, and Temur Kutsia, editors,
  \emph{Proceedings of Algebraic Biology 2007}, volume 4545 of \emph{LNCS},
  pages 66--80. 2007.

\bibitem[Burhenne(1990)]{Burhenne:90a}
Klaus-Dieter Burhenne.
\newblock Implementierung eines {A}lgorithmus zur {Q}uantorenelimination fur
  lineare reelle {P}robleme.
\newblock Diplomarbeit, Universit{\"a}t Passau, Germany, 1990.

\bibitem[Collard(2003)]{Collard:03a}
Jean-Francois Collard.
\newblock \emph{Reasoning About Program Transformations}.
\newblock Springer, 2003.

\bibitem[Dolzmann and Sturm(1996)]{DolzmannSturm:96c}
Andreas Dolzmann and Thomas Sturm.
\newblock Redlog user manual.
\newblock Technical report, FMI, Universit{\"a}t Passau, 1996.

\bibitem[Dolzmann and Sturm(1997{\natexlab{a}})]{DolzmannSturm:97a}
Andreas Dolzmann and Thomas Sturm.
\newblock Redlog: Computer algebra meets computer logic.
\newblock \emph{ACM SIGSAM Bulletin}, 31\penalty0 (2):\penalty0 2--9, June
  1997{\natexlab{a}}.

\bibitem[Dolzmann and Sturm(1997{\natexlab{b}})]{DolzmannSturm:97c}
Andreas Dolzmann and Thomas Sturm.
\newblock Simplification of quantifier-free formulae over ordered fields.
\newblock \emph{Journal of Symbolic Computation}, 24\penalty0 (2):\penalty0
  209--231, 1997{\natexlab{b}}.

\bibitem[Dolzmann et~al.(1998)Dolzmann, Sturm, and Weispfenning]{DSW:98}
Andreas Dolzmann, Thomas Sturm, and Volker Weispfenning.
\newblock A new approach for automatic theorem proving in real geometry.
\newblock \emph{Journal of Automated Reasoning}, 21\penalty0 (3):\penalty0
  357--380, 1998.

\bibitem[Dolzmann et~al.(1999)Dolzmann, Sturm, and
  Weispfenning]{DolzmanEtAl:1999a}
Andreas Dolzmann, Thomas Sturm, and Volker Weispfenning.
\newblock Real quantifier elimination in practice.
\newblock In B.~Heinrich Matzat, Gert-Martin Greuel, and Gerhard Hiss, editors,
  \emph{Algorithmic Algebra and Number Theory}, pages 221--247. Springer, 1999.

\bibitem[{El Kahoui} and Weber(2000)]{El-KahouiWeber:00a}
M'hammed {El Kahoui} and Andreas Weber.
\newblock Deciding {H}opf bifurcations by quantifier elimination in a software
  component architecture.
\newblock \emph{Journal of Symbolic Computation}, 30\penalty0 (2):\penalty0
  161--179, 2000.

\bibitem[Errami et~al.(2011)Errami, Sturm, and Weber]{ErramiSturm:11a}
Hassan Errami, Thomas Sturm, and Andreas Weber.
\newblock Algorithmic aspects of {M}uldowney's extension of the
  {B}endixson-{D}ulac criterion for polynomial vector fields.
\newblock In N.~N. Vassiliev, editor, \emph{Polynomial Computer Algebra}, pages
  25--28, St. Petersburg, Russia, 2011. The Euler International Mathematical
  Institute.

\bibitem[Errami et~al.(2013)Errami, Eiswirth, Grigoriev, Seiler, Sturm, and
  Weber]{ErramiEiswirth:13a}
Hassan Errami, Markus Eiswirth, Dima Grigoriev, Werner~M. Seiler, Thomas Sturm,
  and Andreas Weber.
\newblock Efficient methods to compute {H}opf bifurcations in chemical reaction
  networks using reaction coordinates.
\newblock In \emph{Proceedings of the CASC 2013}, volume 8136 of \emph{LNCS},
  pages 88--99. 2013.

\bibitem[Gatermann et~al.(2005)Gatermann, Eiswirth, and
  Sensse]{GatermannEtAl2005}
Karin Gatermann, Markus Eiswirth, and Anke Sensse.
\newblock Toric ideals and graph theory to analyze {Hopf} bifurcations in mass
  action systems.
\newblock \emph{Journal of Symbolic Computation}, 40\penalty0 (6):\penalty0
  1361--1382, 2005.

\bibitem[Korovin et~al.(2009)Korovin, Tsiskaridze, and
  Voronkov]{KorovinEtAl2009}
Konstantin Korovin, Nestan Tsiskaridze, and Andrei Voronkov.
\newblock Conflict resolution.
\newblock In Ian~P. Gent, editor, \emph{Principles and Practice of Constraint
  Programming -- CP 2009}, volume 5732 of \emph{LNCS}, pages 509--523. 2009.

\bibitem[Korovin et~al.(2014)Korovin, Košta, and Sturm]{KorovinEtAl:2014a}
Konstantin Korovin, Marek Košta, and Thomas Sturm.
\newblock Towards conflict-driven learning for virtual substitution.
\newblock In \emph{Proceedings of the CASC 2014}, volume 8660 of \emph{LNCS},
  pages 256--270. 2014.

\bibitem[Loos and Weispfenning(1993)]{LoosWeispfenning:93a}
R{\"u}diger Loos and Volker Weispfenning.
\newblock Applying linear quantifier elimination.
\newblock \emph{The Computer Journal}, 36\penalty0 (5):\penalty0 450--462,
  1993.

\bibitem[Nelson and Oppen(1979)]{NelsonOppen:1979a}
Greg Nelson and Derek~C. Oppen.
\newblock Simplification by cooperating decision procedures.
\newblock \emph{ACM Transactions on Programming Languages and Systems},
  1\penalty0 (2):\penalty0 245--257, 1979.

\bibitem[Nieuwenhuis et~al.(2006)Nieuwenhuis, Oliveras, and
  Tinelli]{NieuwenhuisEtAl:06a}
Robert Nieuwenhuis, Albert Oliveras, and Cesare Tinelli.
\newblock Solving {SAT} and {SAT} modulo theories: From an abstract
  {D}avis--{P}utnam--{L}ogemann--{L}oveland procedure to {DPLL(T)}.
\newblock \emph{Journal of the ACM}, 53\penalty0 (6):\penalty0 937--977, 2006.

\bibitem[Sturm(1999{\natexlab{a}})]{Sturm1999}
Thomas Sturm.
\newblock Reasoning over networks by symbolic methods.
\newblock \emph{Applicable Algebra in Engineering, Communication and
  Computing}, 10\penalty0 (1):\penalty0 79--96, 1999{\natexlab{a}}.

\bibitem[Sturm(1999{\natexlab{b}})]{Sturm:99a}
Thomas Sturm.
\newblock \emph{Real Quantifier Elimination in Geometry}.
\newblock Doctoral dissertation, Universit{\"a}t Passau, Germany, December
  1999{\natexlab{b}}.

\bibitem[Sturm and Weber(2008)]{SturmWeber:08a}
Thomas Sturm and Andreas Weber.
\newblock Investigating generic methods to solve {H}opf bifurcation problems in
  algebraic biology.
\newblock In K.~Horimoto, editor, \emph{Proceedings of Algebraic Biology 2008},
  volume 5147 of \emph{LNCS}, pages 200--215. 2008.

\bibitem[Sturm and Weispfenning(1997)]{SturmWeispfenning1997}
Thomas Sturm and Volker Weispfenning.
\newblock Computational geometry problems in {REDLOG}.
\newblock In Dongming Wang, editor, \emph{Automated Deduction in Geometry},
  volume 1360 of \emph{LNCS}, pages 58--86. 1997.

\bibitem[Sturm et~al.(2009)Sturm, Weber, Abdel-Rahman, and {El
  Kahoui}]{SturmWeber:09a}
Thomas Sturm, Andreas Weber, Essam~O. Abdel-Rahman, and M'hammed {El Kahoui}.
\newblock Investigating algebraic and logical algorithms to solve {H}opf
  bifurcation problems in algebraic biology.
\newblock \emph{Mathematics in Computer Science}, 2\penalty0 (3):\penalty0
  493--515, 2009.

\bibitem[Weber et~al.(2011)Weber, Sturm, and Abdel-Rahman]{WeberSturm:11a}
Andreas Weber, Thomas Sturm, and Essam~O. Abdel-Rahman.
\newblock Algorithmic global criteria for excluding oscillations.
\newblock \emph{Bull. Math. Biol.}, 73\penalty0 (4):\penalty0 899--916, April
  2011.

\bibitem[Weispfenning(1988)]{Weispfenning:88a}
Volker Weispfenning.
\newblock The complexity of linear problems in fields.
\newblock \emph{Journal of Symbolic Computation}, 5\penalty0 (1\&2):\penalty0
  3--27, 1988.

\bibitem[Weispfenning(1994{\natexlab{a}})]{Weispfenning:1994a}
Volker Weispfenning.
\newblock Quantifier elimination for real algebra---the cubic case.
\newblock In \emph{Proceedings of the ISSAC 1994}, pages 258--263. ACM Press,
  New York, 1994{\natexlab{a}}.

\bibitem[Weispfenning(1994{\natexlab{b}})]{Weispfenning:94b}
Volker Weispfenning.
\newblock Parametric linear and quadratic optimization by elimination.
\newblock Technical Report MIP-9404, Universit{\"a}t Passau, Germany,
  1994{\natexlab{b}}.

\bibitem[Weispfenning(1997{\natexlab{a}})]{Weispfenning:97b}
Volker Weispfenning.
\newblock Quantifier elimination for real algebra---the quadratic case and
  beyond.
\newblock \emph{Applicable Algebra in Engineering Communication and Computing},
  8\penalty0 (2):\penalty0 85--101, 1997{\natexlab{a}}.

\bibitem[Weispfenning(1997{\natexlab{b}})]{Weispfenning:97d}
Volker Weispfenning.
\newblock Simulation and optimization by quantifier elimination.
\newblock \emph{Journal of Symbolic Computation}, 24\penalty0 (2):\penalty0
  189--208, 1997{\natexlab{b}}.

\bibitem[Weispfenning(2001)]{Weispfenning2001}
Volker Weispfenning.
\newblock Semilinear motion planning in {REDLOG}.
\newblock \emph{Applicable Algebra in Engineering, Communication and
  Computing}, 12\penalty0 (6):\penalty0 455--475, 2001.

\end{thebibliography}
\end{document}